\documentclass[11pt]{article}

\usepackage{dsfont}

\usepackage{epsfig, color}
\usepackage{color,graphicx}
\usepackage{amsfonts}
\usepackage{color}
\usepackage{amsmath}
\allowdisplaybreaks[4]
\usepackage{amssymb}
\usepackage{amsthm}
\usepackage{natbib}
\usepackage{fancybox}
\usepackage{ulem}
\usepackage{subfigure}
\usepackage{multirow}
\usepackage{bbm}
\usepackage{enumerate}
\usepackage{caption}
\usepackage{subfigure}
\usepackage{leftidx}

\newtheorem{theorem}{\hspace{6mm}Theorem}[section]

\newtheorem{rem}[theorem]{\hspace{6mm}Remark}
\newtheorem{exam}[theorem]{\hspace{6mm}Example}
\newtheorem{cor}[theorem]{\hspace{6mm}Corollary}

\def\be{\begin{equation}}
\def\ee{\end{equation}}
\numberwithin{equation}{section}

\begin{document}

\title{Dual control Monte Carlo method for tight bounds of value function under Heston stochastic volatility model} 

\vskip 4mm

\author{Jingtang Ma\thanks{School of Economic Mathematics, Southwestern University of
Finance and Economics, Chengdu, 611130, China (Email: mjt@swufe.edu.cn). The work was supported by National Natural Science Foundation of China (Grant No.
11671323) and Program for New Century Excellent Talents in University (China Grant No. NCET-12-0922).}, Wenyuan Li\thanks{School of Economic Mathematics, Southwestern University of Finance and Economics, Chengdu, 611130, China (Email: Wylfm@2011.swufe.edu.cn).} and Harry Zheng
\thanks{Corresponding author. Department of Mathematics, Imperial College, London SW7 2BZ, UK
(Email: h.zheng@imperial.ac.uk).}}

\date{}

\maketitle

\begin{abstract}
The aim of this paper is to study the fast computation of the lower and upper bounds on the value function for  utility maximization under the Heston stochastic volatility model with general utility functions. It is well known there is a closed form solution of the HJB equation for power utility due to its homothetic property. It is not possible to get closed form solution for general utilities  and there is little literature on the numerical scheme to solve the HJB equation for the Heston model. In this paper we propose an efficient dual control Monte Carlo method for computing tight lower and upper bounds of the value function. We identify  a particular form of the dual control which leads to the closed form upper bound for a class of utility functions, including power, non-HARA and Yarri utilities. Finally, we perform some numerical tests to see the efficiency, accuracy, and robustness
of the method. The numerical results support strongly our proposed scheme.

\end{abstract}

\vspace{1.0cm}


\smallskip


\section{Introduction}

Dynamic portfolio optimization is one of most studied research areas in mathematical finance. Stochastic control and convex duality are two standard methods to solve utility maximization problems. For a complete market such as the Black-Scholes model, the problem has already been solved. One may first find the optimal terminal wealth and then use the martingale representation theorem to find the optimal control, or one may solve the HJB equation to find the optimal value function  and optimal control, see many excellent books for expositions, e.g., \cite{Shreve1998}, \cite{Pham2009}.

For an incomplete market driven by some Markovian processes, one may use the dynamic programming equation to solve the problem. One well-known example is the  Heston stochastic volatility model. The HJB equation has two state variables (wealth and variance). For a power utility, one may decompose the solution to reduce the dimensionality of state variables by one and get a simplified nonlinear PDE with one state variable (variance). Thanks to the affine structure of the Heston model,  \cite{Zariphopoulou2001} uncovers a clever transformation that simplifies the nonlinear  equation further into an equivalent linear PDE and  derives a closed-form  solution, see \cite{Zariphopoulou2001} and \cite{Kraft2005} for details.  \cite{Kallsen2010} extend the Heston model to general affine stochastic volatility models and \cite{Richter2014} to multi-dimensional affine jump-diffusion stochastic volatility models with the martingale method and matrix Riccati ordinary differential equations.

The success of finding a closed-form solution for utility maximization with the Heston model (or general affine stochastic processes) crucially depends on the underlying utility being a power utility if the wealth process is exponential (or exponential utility if the wealth process is additive). Such combination of utility and wealth process would decouple wealth and variance variables in the optimal value function and the special affine structure of variance would help to give a closed-form solution, whether using the HJB equation or the quadratic backward stochastic differential equation. For general utilities, there is no way one can decouple wealth and variance variables and, consequently, there are no results for the existence of a classical solution to the HJB equation, let alone a closed form solution. Furthermore, due to high nonlinearity of the HJB equation with two state variables, it is also very difficult to find an efficient numerical method to solve it.

For a Black-Scholes market with closed convex cone constraints for controls and general continuous concave utility functions,
\cite{Bian2015} (see also \cite{Bian2011}) show that there exists a classical solution to the HJB equation and the solution has a representation in terms of a solution to the dual HJB equation. That approach does not work  for an incomplete market model such as the Heston model. The reason is that the stochastic volatility is not a traded asset (unless an additional volatility related security is introduced) and the Heston model is not a geometric Brownian motion and the dual HJB equation is an equally difficult nonlinear PDE with two state variables.

Although the dual control method cannot solve general utility maximization problems with incomplete market models, it nevertheless provides the valuable information for the optimal value function. The dual value function supplies a natural upper bound for the original primal value function due to the  dual relation and a feasible control which may be used to provide a good lower bound for the primal value function. If one can make the gap between the lower and upper bounds small, then one can at least find an approximate solution to the primal value function, which would be impossible without using the dual control method. This idea has been applied successfully to find the approximate optimal value function for regime switching asset price models with general utility functions, see  \cite{Ma2017}.

In this paper we adopt this line of attack to utility maximization with the Heston stochastic volatility model.
We derive the dual control problem and recover the optimal solution  for power utility in  \cite{Zariphopoulou2001} and \cite{Kraft2005}. For general utilities, we propose a Monte Carlo method to compute the lower and upper bounds for the primal value function. Some upper bounds can be computed efficiently with the closed form formula or the fast Fourier-cosine method thanks to the affine structure of the Heston model. Numerical tests for power, non-HARA  and Yarri utilities show that  these bounds are tight, which  provides a good approximation to the primal value function. To the best of our knowledge, this is the first time an efficient dual control  Monte Carlo method is proposed to find the tight lower and upper bounds for the value function with the Heston stochastic volatility model and general utility functions.

The rest of the paper is arranged as follows. Section \ref{sec:dual-control} discusses  the dual control method and derives the same closed-form solution for power utility as that in \cite{Kraft2005}.  Section \ref{section:lower-upper-bound} presents the dual control  Monte Carlo method for computing tight  lower and upper bounds of the value function. Section \ref{sec:examples} provides the closed-form upper bound for a specific  form of the dual control and a class of utility functions, including power, non-HARA and Yarri utilities. Section \ref{sec:numerical-test} performs numerical tests to see  the efficiency, accuracy, and robustness of the method. Section \ref{conclusions} concludes. Appendix A gives the closed form solution of the Riccati equation associated with the Heston model and Appendix B explains the COS method for computing the upper bound with Yarri utility.

\section{ The Heston model and the dual control method}\label{sec:dual-control}

Assume that $(\Omega, {\cal F}, {\cal F}_t, P)$  is a given probability space with filtration ${\cal F}_t$ generated by standard Brownian motions $W^s$ and $W^v$ with correlation coefficient $\rho$ and
completed with all $P$-null sets. The market is composed of two traded assets, one savings account $B$ with riskless interest rate $r$ and one risky asset $S$ satisfying a stochastic differential equation (SDE) (see \cite{Heston1993}):
\begin{eqnarray*}
dS_t = S_t[ (r+A v_t) dt + \sqrt{v_t} dW_t^s],
\end{eqnarray*}
 where $A$ is a constant representing the market price of risk, $v$ is an asset variance process satisfying
a mean-reverting square-root process:
\[
 dv_t =\kappa(\theta-v_t)dt + \xi \sqrt{v_t} dW_t^v,
\]
 $\theta$ is the long-run average volatility, $\kappa$ the rate that $v_t$ reverts to $\theta$,  $\xi$ the variance of $\sqrt{v_t}$, and all parameters are positive constants and satisfy the Feller condition $2\kappa\theta \geq \xi^2$ to ensure that $v_t$ is strictly positive.

 Let $X$ be the wealth process. At time $t\in [0,T]$ the investor allocates a proportion $\pi_t$ of  wealth $X_t$ in risky asset $S$ and the remaining wealth in savings account $B$. Then the wealth process $X$  satisfies the SDE:
\begin{equation}\label{X-process}
dX_t =  X_t[(r + \pi_t  A v_t) dt + \pi_t \sqrt{v_t}dW_t^s],
\end{equation}
where $\pi$ is a progressively measurable control process.

The utility maximization problem is defined by
\begin{equation} \label{primal}
\sup_{\pi} E[U(X_T)]\mbox{ subject to (\ref{X-process})},
\end{equation}
where $U$ is a utility function that is continuous, increasing and concave (but not necessarily strictly increasing and strictly concave) on $[0,\infty)$, and $U(0)=0$.
To solve \eqref{primal} with the stochastic control method, we define the value function
 \begin{equation}\label{optimization-W}
 \mathcal{W}(t,x,v):=\sup_{\pi\in\Pi_{t}}E_{t,x,v}[U(X_T)],
\end{equation}
where  $E_{t,x,v}$ is the conditional expectation operator given $X_t=x$ and $v_t=v$, and
  $\Pi_{t}:=\{\pi_s,\, s\in[t,T]\}$ is the set of all admissible control strategies over $[t,T]$.

 By the dynamic programming principle,  $ \mathcal{W}$ satisfies the following HJB equation:
\begin{equation}\label{primalHJB}
{\partial \mathcal{W}\over \partial t}
+ \sup_\pi\left\{(rx+\pi x Av)\mathcal{W}_x + \kappa(\theta-v)\mathcal{W}_v
+ {1\over 2} \pi^2 x^2 v \mathcal{W}_{xx} + {1\over 2} \xi^2 v \mathcal{W}_{vv}
+ \rho \pi x \xi v \mathcal{W}_{xv} \right\}=0
\end{equation}
with the terminal condition $\mathcal{W}(T,x,v)=U(x)$, where $\mathcal{W}_x$
is the partial derivative of  $\mathcal{W}$ with respect to $x$ and evaluated at $(t,x,v)$,  the other derivatives are similarly defined.  The maximum in (\ref{primalHJB}) is achieved at
\begin{equation}\label{pi}
\pi = -\frac{A\mathcal{W}_x}{x\mathcal{W}_{xx}}-\frac{\xi\rho\mathcal{W}_{xv}}{x\mathcal{W}_{xx}}.
\end{equation}
Inserting (\ref{pi}) into (\ref{primalHJB}) gives a nonlinear PDE
\begin{equation}\label{X-HJB}
{\partial \mathcal{W}\over \partial t}+rx\mathcal{W}_x + \kappa(\theta-v)\mathcal{W}_v + \frac{1}{2}\xi^2 v\mathcal{W}_{vv}-\frac{1}{2\mathcal{W}_{xx}}[A\sqrt{v}\mathcal{W}_x
+\xi\rho\sqrt{v}\mathcal{W}_{xv}]^2 = 0.
\end{equation}
For a power utility $U(x)=(1/p)x^p$, $0<p<1$,
the solution of (\ref{X-HJB}) can be decomposed as
$$\mathcal{W}(t,x,v)=U(x) f(t,v)$$
 for some function $f$ which satisfies
\begin{equation}
\label{f-HJB}
{\partial  f\over \partial t}+pr  f + \kappa(\theta-v) f_v + \frac{1}{2}\xi^2 v f_{vv}-\frac{pv}{2(p-1) f}[A f
+\xi\rho f_{v}]^2 = 0
\end{equation}
with the termianl condition $f(T,v)=1$. The optimal control is given by
\begin{equation*}\label{control1.8}
 \pi = -{A\over (p-1)} - {\xi\rho  f_v \over (p-1) f}.
 \end{equation*}
The equation (\ref{f-HJB}) is simpler than the equation (\ref{X-HJB}) but is still a nonlinear PDE.  \cite{Zariphopoulou2001} suggests a  clever transformation
$$  f(t,v)=\hat f(t,v)^{\varsigma}$$
with $\varsigma=(1-p)/(1-p+\rho^2 p)$, which removes the nonlinear terms in (\ref{f-HJB}),  and $\hat  f$ satisfies a linear PDE
\begin{equation}\label{tildef-HJB}
{\partial \hat  f\over \partial t}+{pr\over \varsigma}\hat  f + \kappa(\theta-v)\hat  f_v + \frac{1}{2}\xi^2 v \hat  f_{vv}-\frac{p}{2(p-1)\varsigma}(A^2v\hat  f + 2A\xi\rho v\varsigma \hat  f_v) = 0
\end{equation}
with the terminal condition $\hat  f(T,v)=1$. The equation (\ref{tildef-HJB}) can be easily solved and the solution $\hat  f$ has a Feynman-Kac representation. In fact, thanks to the affine structure of the Heston model, the solution $f$ of the equation (\ref{f-HJB}) has an analytical form as
$$ f(t,v)=\exp(C(t)+D(t)v),$$
where $C$ and $D$ are solutions of
some Riccati-type ODEs with  terminal conditions $C(T)=0$ and $D(T)=0$ and can be easily solved, and the optimal control is given by
$\pi=(A + \xi\rho D(t))/(1-p)$,
 see \cite{Kraft2005} for details.

The success of simplifying the HJB equation (\ref{X-HJB}) to a solvable nonlinear PDE (\ref{f-HJB}) crucially depends on the assumption that the utility function is a power utility.
For general utility functions  (e.g., non-HARA and Yarri utilities), it is virtually impossible one can find the analytical solutions.

The dual function of $U$ is defined by
\begin{equation}\label{dual-func}
\widetilde{U}(y)=\sup_{x\geq 0}[U(x)-xy],
\end{equation}
for $y\geq 0$. The function $\widetilde{U}(y)$ is a continuous, decreasing and convex function on $[0,\infty)$ and satisfies $\widetilde{U}(\infty)=0$.
Suppose a dual process of the following form
\begin{equation*}
d Y_t = Y_t [\alpha_t dt + \beta_t dW_t^s + \gamma_t dW_t^v]
\end{equation*}
with the initial condition $Y_0=y$. If $XY$ is a super-martingale for any control process $\pi$, then
$$ E[ U(X_T)] \leq E[\widetilde{U}(Y_T)] + xy,$$
where $X_0=x$ is the initial wealth, which leads to
$$ \sup_{\pi} E[ U(X_T)] \leq \inf_{y}( \inf_{\alpha,\beta,\gamma}E[\widetilde{U}(Y_T)] + xy),$$
 and we have a weak duality relation. To make $XY$ a super-martingale, we can use It\^{o}'s formula to get
\begin{equation*}
\alpha_t \leq -r,~ \beta_t = -A\sqrt{v_t} - \rho \gamma_t.
\end{equation*}
Furthermore, since $\widetilde{U}$ is a decreasing convex function, we must have
$\alpha_t=-r$.
Therefore, the dual process is given by
\begin{equation}\label{Y-process}
dY_t = Y_t[-r dt - (\rho \gamma_t + A\sqrt{v_t})dW_t^s + \gamma_t dW_t^v]
\end{equation}
with the initial condition $Y_0=y$, where $\gamma$ is a dual control process and $y$ is also a dual control variable.
The solution to \eqref{Y-process} at time $T$, with initial condition $Y_{t}=y$, can be written as
\begin{align*}\label{Y-expression}
Y_T = y\exp(M_{t,T}),
\end{align*}
where
\begin{equation*}\label{M-t-T}
M_{t,T} = -\int_t^T\left(r + \frac{1}{2}(1-\rho^2)\gamma_u^2 + \frac{1}{2}A^2 v_u \right) du - \int_t^T (\rho \gamma_u + A \sqrt{v_u}) dW^s_u + \int_t^T \gamma_u dW^v_u.
\end{equation*}
Define the dual value function as
\begin{equation*}\label{dual-val-func}
\widetilde{\mathcal{W}}(t,y,v):=\inf_{\gamma}E_{t,y,v}[\widetilde{U}(Y_T)].
\end{equation*}
By the dynamic programming principle, $\widetilde{\mathcal{W}}$ satisfies the following dual HJB equation
\begin{eqnarray}\label{Y-HJB-1}
 && {\partial \widetilde{\mathcal{W}}\over \partial t}
+ \inf_{\gamma}\bigg\{ - ry\widetilde{\mathcal{W}}_y + \kappa(\theta-v)\widetilde{\mathcal{W}}_v +\frac{1}{2}y^2[A^2v +\gamma^2(1-\rho^2)]\widetilde{\mathcal{W}}_{yy} \nonumber\\
 &&\qquad +\, y[\gamma\xi\sqrt{v}(1-\rho^2)-Av\xi\rho]\widetilde{\mathcal{W}}_{yv} +\frac{1}{2}\xi^2 v\widetilde{\mathcal{W}}_{vv}\bigg\}=0
\end{eqnarray}
with the terminal condition $\widetilde{\mathcal{W}}(T,y,v)=\widetilde{U}(y)$.
The minimum in  (\ref{Y-HJB-1}) is achieved at
\begin{equation}\label{find-gammastar-2}
\gamma = -\frac{\xi\sqrt{v}\widetilde{\mathcal{W}}_{yv}}{y\widetilde{\mathcal{W}}_{yy}}.
\end{equation}
Inserting \eqref{find-gammastar-2} into \eqref{Y-HJB-1} gives
\begin{equation}\label{Y-HJB-2}
 {\partial \widetilde{\mathcal{W}}\over \partial t} - ry\widetilde{\mathcal{W}}_y + \kappa(\theta-v)\widetilde{\mathcal{W}}_v +
\frac{1}{2}\widetilde{\mathcal{W}}_{yy}y^2A^2v - \frac{1}{2}\frac{\xi^2 v \widetilde{\mathcal{W}}_{yv}^2}{\widetilde{\mathcal{W}}_{yy}}(1-\rho^2) 
- Av\xi\rho y\widetilde{\mathcal{W}}_{yv} + \frac{1}{2}\widetilde{\mathcal{W}}_{vv} \xi^2 v
=0.
\end{equation}

\begin{theorem}\label{complete-dual-transform}
Let $\widetilde{\mathcal{W}}\in C^{1,2,2}$ be the solution of the dual HJB equation \eqref{Y-HJB-1}
and  let $\widetilde{\mathcal{W}}$ be strictly convex in $y$ and satisfy $\widetilde{\mathcal{W}}_y(t,0,v)=-\infty$ and $\widetilde{\mathcal{W}}_y(t,\infty,v)=0$. Then the primal value function is given by
\begin{equation*}\label{complete-inverse}
\mathcal{W}(t,x,v) = 
\widetilde{\mathcal{W}}(t,y^*,v) + xy^*,
\end{equation*}
where $y^*=y(t,x,v)$ is the solution of the equation
$$ \widetilde{\mathcal{W}}_y(t,y,v)  +x =0.$$
Furthermore, $\mathcal{W}\in C^{1,2,2}$
is the solution of the HJB equation \eqref{X-HJB} with the boundary condition $\mathcal{W}(T,x,v)=U(x)$ and the optimal feedback control is given by
$$ \pi(t,x,v) = {A\over x}y^*  \widetilde{\mathcal{W}}_{yy}(t,y^*,v) + {\xi\rho\over x} \widetilde{\mathcal{W}}_{yv}(t,y^*,v).
$$
\end{theorem}
\begin{proof}
Define
\begin{equation*}\label{dual-relation}
\widehat{\mathcal{W}}(t,x,v)=\inf\limits_{y>0}\left[\widetilde{\mathcal{W}}(t,y,v)+xy\right].
\end{equation*}
Since $\widetilde{\mathcal{W}}\in C^{1,2,2}$ and is strictly convex in $y$ and $\widetilde{\mathcal{W}}_y(t,0,v)=-\infty$ and $\widetilde{\mathcal{W}}_y(t,\infty,v)=0$, we have
\begin{equation*}
\widehat{\mathcal{W}}(t,x,v)=\widetilde{\mathcal{W}}(t,y^*,v)+xy^*,
\end{equation*}
where $y^*=y(t,x,v)$ satisfies
$ \widetilde{\mathcal{W}}_y(t,y,v)+ x = 0$.
Using the Implicit Function Theorem, we have  $y\in C^{1,2,2}$ and therefore $\widehat{\mathcal{W}}\in C^{1,2,2}$.
Simple calculus shows that
$$
\frac{\partial \widehat{\mathcal{W}}}{\partial t}
= \frac{\partial\widetilde{\mathcal{W}}}{\partial t},\quad
\widehat{\mathcal{W}}_x = y,\quad
\widehat{\mathcal{W}}_v = \widetilde{\mathcal{W}}_v,
$$
and
$$
\widehat{\mathcal{W}}_{xx} = -{1\over \widetilde{\mathcal{W}}_{yy}},\quad
\widehat{\mathcal{W}}_{xv} = -{\widetilde{\mathcal{W}}_{yv} \over
\widetilde{\mathcal{W}}_{yy}},\quad
\widehat{\mathcal{W}}_{vv} = -{(\widetilde{\mathcal{W}}_{yv})^2\over \widetilde{\mathcal{W}}_{yy}} + \widetilde{\mathcal{W}}_{vv}.
$$
Substituting these relations into \eqref{Y-HJB-2} gives that $\widehat{\mathcal{W}}$ satisfies the HJB equation \eqref{X-HJB}. Moreover it follows from the conjugate equation \eqref{dual-relation} and $\widetilde{W}(T,y,v)=\widetilde{U}(y)$ that $\widehat{\mathcal{W}}(T,x,v)=U(x)$.
The verification theorem then gives $\mathcal{W}(t,x,v)=\widehat{\mathcal{W}}(t,x,v)$. The optimal feedback control is derived from (\ref{pi}) and the dual relations of the derivatives.
\end{proof}

\begin{rem} Theorem \ref{complete-dual-transform} shows there is no duality gap if  the dual control $\gamma$ takes the form \eqref{find-gammastar-2}. This is interesting in theory and is useful if one knows $\widetilde{\mathcal{W}}$. In general, it is highly unlikely one can find $\widetilde{\mathcal{W}}$, which requires to solve an equally difficult nonlinear PDE \eqref{Y-HJB-2}. However, if one can choose a dual control $\gamma$ which gives a good approximation to $\widetilde{\mathcal{W}}$, then one may follow  Theorem \ref{complete-dual-transform} to get good approximations to the primal value function and the optimal feedback control. This is essentially the idea we use to design a Monte Carlo method for computing tight lower and upper bounds of  the primal value function in the next section. For power utility $U(x)=\frac{x^p}{p}$, we can find a closed-form solution of \eqref{Y-HJB-2} and therefore solve the primal problem with the dual method. This is explained in the next result.
\end{rem}

\begin{cor}\label{power-primal-analytical}
For power utility $U(x) = \frac{x^p}{p}$, Theorem \ref{complete-dual-transform} gives the closed-form formula for the primal value of \eqref{optimization-W} as follows:
\begin{equation}\label{power-solution-1}
\mathcal{W}(t,x,v) = \frac{x^p}{p}\exp\Big[(1-p)(C(t)+D(t)v) \Big],
\end{equation}
where $C(t)$ and $D(t)$ are  given by \eqref{Ct-2} and \eqref{Dt-2}, respectively, with $\underline{t}=0$, $\bar{t}=T$, $f_1=0$, $f_2=0$,
$d_1=-\kappa\theta,\, d_2=rp/(p-1)$, and
$$a=\frac{1}{2}\xi^2[p(1-\rho^2)-1],\quad b=\kappa-A\xi\rho\frac{p}{1-p},\quad \eta =-\frac{1}{2}\frac{p}{(1-p)^2}A^2,$$
provided $b^2-4a\eta>0$ and $\frac{m_1}{m_2} \notin [e^{-k_1T},1]$. The optimal control at time $t$ is given by
$$ \pi_t= {A \over 1-p} + \xi\rho D(t).
$$
\end{cor}
\begin{proof}
The dual function of $U$ is given by $\widetilde{U}(y)=-(1/q)y^q$, where
$q=p/(p-1)$. We may set $\widetilde{\mathcal{W}}(t,y,v)=\widetilde{U}(y)\tilde f(t,v)$ and substitute it into the equation (\ref{Y-HJB-2}) to get a simplified equation for $\tilde  f$:
\begin{equation}\label{e28}
{\partial \tilde  f\over \partial t} - rq \tilde  f + \kappa(\theta -v) \tilde  f_v
+{1\over 2}q(q-1)\tilde  f A^2 v - {q\over 2(q-1)} \xi^2v {\tilde  f_v^2\over \tilde  f}(1-\rho^2) -\, qAv\xi \rho \tilde  f_v + {1\over 2}\tilde  f_{vv}\xi^2 v =0
\end{equation}
with the terminal condition $\tilde  f(T,v)=1$.
We can solve equation (\ref{e28}) by setting $\tilde  f(t,v)=\exp(C(t)+D(t)v)$ and plugging $\tilde  f$ into (\ref{e28}) to get two ODEs for $C$ and $D$ as follows:
\begin{equation*}
\left\{
\begin{array}{l}
C'(t)=-\kappa\theta D(t) - r\frac{p}{1-p},\\
C(T)=0,
\end{array}\right.
\end{equation*}
and
\begin{equation*}
\left\{
\begin{array}{l}
D'(t)=D^2(t)\frac{1}{2}\xi^2[p(1-\rho^2)-1] + D(t)(\kappa-A\xi\rho\frac{p}{1-p})-\frac{1}{2}\frac{p}{(1-p)^2}A^2,\\
D(T)=0.
\end{array}\right.
\end{equation*}
 We can easily find $C(t)$ once $D(t)$ is known and  solve the Riccati equation to get  a closed-form solution $D(t)$, see Appendix.
Next  we solve the  equation $\widetilde{\mathcal{W}}_y + x = 0$ to get
\[
y^*=\left[x\exp(-C(t)-D(t)v)\right]^{p-1}.
\]
Using Theorem \ref{complete-dual-transform}, we obtain the primal value function by
\begin{equation*}\label{power-primal-value-2}
\mathcal{W}(t,x,v) = \widetilde{W}(t,y^*,v)+xy^* =\frac{x^p}{p}\exp[(1-p)(C(t)+D(t)v)].
\end{equation*}
\end{proof}

\begin{rem}
Corollary \ref{power-primal-analytical} shows that the dual control method of Theorem \ref{complete-dual-transform} gives the closed-form formula for the primal value function with the power utility.
After communicating the notations, we see that formula \eqref{power-solution-1} is the same as that in Prop 5.2 of \cite{Kraft2005}.
\end{rem}

 \section{Monte Carlo lower and upper bounds}\label{section:lower-upper-bound}
For general utility functions, it seems impossible we can solve the primal problem by using Theorem \ref{complete-dual-transform} as the dual problem is equally difficult.  Note that
\begin{equation}\label{Wuppperbound}
\sup_{\pi} E[ U(X_T)] \leq \inf_{y}( \inf_{\gamma}E[\widetilde{U}(Y_T)] + xy)
\leq \inf_{y}( E[\widetilde{U}(Y_T)] + xy),
\end{equation}
for all dual controls $\gamma$. For every fixed $\gamma$, define
\begin{equation}\label{new-dual-val-function}
\mathcal{Z}(t,y,v) = E_{t,y,v}[\widetilde{U}(Y_T)].
\end{equation}
Then $\mathcal{Z}$ is an upper bound and can be easily computed with simulation. Note that $\mathcal{Z}(t,y,v)$ depends on the choice of dual control $\gamma$.
Denote the conjugate function of $\mathcal{Z}(t,y,v)$ for fixed $t$ and $v$ by
\begin{equation}\label{Wuppperbound-C}
\overline{\mathcal{W}}(t,x,v)=\inf_{y>0} [\mathcal{Z}(t,y,v) + xy].
\end{equation}
The following theorem presents the tight lower and upper bounds on the primal value function.
\begin{theorem}
Let $\mathcal{S}$ be a set of admissible dual controls and $\overline{\mathcal{W}}(t,x,v)$ be given by \eqref{Wuppperbound-C}.
Then the optimal value function $\mathcal{W}(t,x,v)$ defined in \eqref{optimization-W} satisfies
\begin{equation}
 \mathcal{W}(t,x,v) \leq \inf_{\gamma \in \mathcal{S}}  \overline{\mathcal{W}}(t,x,v). \label{tightUB-C}
 \end{equation}
Furthermore,  assume that $\mathcal{Z}(t,y,v)$ given by $\eqref{new-dual-val-function}$ is twice continuously differentiable and strictly convex for $y>0$ with fixed $t$ and $v$, $y^* = y(t,x,v;\gamma)$ is the solution of the equation
\begin{equation}\label{nonlinear-algebraic-eq}
 \mathcal{Z}_y(t,y,v) + x = 0,
\end{equation}
the feedback control $\bar{\pi}(t,x,v)$, defined by
\begin{equation}\label{pi-expression}
\bar{\pi}(t,x,v)
:= A\frac{y^*}{x}\mathcal{Z}_{yy}(t,y^{*},v) - \frac{\xi\rho}{x}\mathcal{Z}_{yv}(t,y^{*},v),
\end{equation}
 is admissible, and $\bar{X}$ is the unique strong solution of SDE \eqref{X-process} with the feedback control  $\pi_t = \bar{\pi}(t,\bar{X}_t,v_t)$ for $t\in [0,T]$. Define
\begin{equation}\label{Wlowerbound-C}
\underline{\mathcal{W}}(t,x,v):= E_{t,x,v}[U(\bar{X}_{T})].
\end{equation}
Then the optimal value function $\mathcal{W}(t,x,v)$ satisfies
\begin{equation}
 \mathcal{W}(t,x,v) \geq \sup_{\gamma \in \mathcal{S}}  \underline{\mathcal{W}}(t,x,v). \label{tightLB-C}
\end{equation}
\end{theorem}
\begin{proof}
It is obvious from \eqref{Wuppperbound} and the definitions of $\overline{\mathcal{W}}(t,x,v)$ and
$\underline{\mathcal{W}}(t,x,v)$.
\end{proof}

\begin{rem}
Clearly, if $\mathcal{S} \subset \widetilde{\mathcal{S}}$, then
\begin{equation}
\mathcal{W}(t,x,v)\leq \inf_{\gamma \in \widetilde{\mathcal{S}}}\overline{\mathcal{W}}(t,x,v) \leq \inf_{\gamma \in \mathcal{S}}\overline{\mathcal{W}}(t,x,v)
\end{equation}
Using $\widetilde{\mathcal{S}}$ instead of $\mathcal{S}$ gives a tighter upper bound but is more expensive in computation. The same applies to the lower bound. For numerical tests in Section \ref{sec:numerical-test}, we  choose the set $\mathcal{S}$ to contain the following  dual controls: $\gamma_{t}=c(t),\, c(t)\sqrt{v_{t}},\, c(t)v_{t}$, where $c$ is a piecewise constant function
 \begin{equation}\label{piecewise-constant-c}
  c(t)=\sum_{j=1}^n c_j 1_{(t_{j-1},t_j]}(t),
 \end{equation}
with $0=t_0<t_1<\ldots<t_{n-1}<t_n=T$ for $n\geq 1$ and $c_j$, $j=1,\ldots,n$ being arbitrary constants.
\end{rem}

\begin{rem} \label{rk3.2}
If the dual function $\widetilde{U}$ has the form
$
 \widetilde{U}(y) = \sum_{i=1}^K  \widetilde{U}_i(y),
$
where $\widetilde{U}_i(y)=-(1/q_i)y^{q_i}$ for $q_i<0$ and $i=1,\ldots,K$, then
$$\mathcal{Z}(t,y,v) = \sum_{i=1}^K  \widetilde{U}_i(y) F_i(t,v),$$
where $F_i(t,v)=E_{t,v}[\exp(q_iM_{t,T})]$. The upper bound is given by
$$  \overline{\mathcal{W}}(t,x,v) = \mathcal{Z}(t,y^*,v) + xy^*$$
and $y^*$ is the solution of equation (\ref{nonlinear-algebraic-eq}):
$ - \sum_{i=1}^K y^{q_i-1} F_i(t,v) + x =0. $
The feedback control for the lower bound is given by (\ref{pi-expression}):
$$
x\bar{\pi}(t,x,v)
= \sum_{i=1}^K (y^*)^{q_i-1}
\left(A(1- q_i) F_i(t,v)
+\xi\rho   {\partial \over \partial v}F_i(t,v)\right).
$$
For fixed dual control $\gamma_t$, $0\leq t\leq T$, we can use the Monte Carlo method  to compute $F_i(t,v)$  and approximate $ {\partial \over \partial v}F_i(t,v)$ with the finite difference
$(F_i(t,v+h)-F_i(t,v-h))/(2h)$ for  sufficiently small $h>0$. If $K=1$, we have a closed-form solution $y^*$. If $K>1$, we can use the Newton-Raphson method to find $y^*$.
\end{rem}

\begin{rem} \label{rk3.3}
If the dual function $\widetilde{U}$ is Lipschitz continuous, then we may use the pathwise differentiation method to compute $\mathcal{Z}_y(t,y,v)$, that is,
$$ \mathcal{Z}_y(t,y,v) = E_{t,y,v}[\widetilde{U}'(y\exp(M_{t,T})) \exp(M_{t,T}) ].$$
For example, the dual function of Yarri utility (see \eqref{threshold-utility}) is given by
$\widetilde{U}(y)=L(1-y)^+$, we have $\widetilde{U}'(y)=-L1_{\{y<1\}}$, where $1_S$ is an indicator which equals 1 if  $S$ happens and 0 otherwise. We can then approximate $\mathcal{Z}_{yy}(t,y,v) $  and $\mathcal{Z}_{yv}(t,y,v) $ with finite differences
$( \mathcal{Z}_y(t,y+h,v) - \mathcal{Z}_y(t,y-h,v) )/(2h)$ and $( \mathcal{Z}_y(t,y,v+h) - \mathcal{Z}_y(t,y,v-h) )/(2h)$, respectively, for sufficiently small $h>0$.

\end{rem}

The Monte-Carlo methods can be used to find the tight lower and upper bounds, analogously to the algorithm developed by \cite{Ma2017}. To implement the method, we need to discretize the dual process $Y$ in \eqref{Y-process}.

Although under the Feller condition $2\kappa\theta \geq \xi^2$, the stochastic volatility process $v_t$ is strictly positive, the discretization of the SDE will draw $v_t$ below zero. To deal with this situation, we apply the full-truncation Euler method first proposed in \cite{Lord2010}, which outperforms many biased schemes in terms of bias and root-mean-squared error. The stochastic volatility $v_t$ can be discretized in the following form
\begin{equation*}\label{vt-discretization}
v_{t+\Delta t} = v_t + \kappa(\theta-v_t^+) \Delta t + \xi \sqrt{v_t^+} \sqrt{\Delta t} Z_1,
\end{equation*}
where $v_t^+ = \max(0,v_t)$ and $Z_1$ is a standard normal variate,
the processes $Y_t$ and $X_t$ by the Euler method,
\begin{eqnarray*}\label{Yt-discretization}
Y_{t+\Delta t} &=& Y_t  \exp\left[\left(-r +\frac{1}{2}(\rho^2-1)\gamma_t^2 - \frac{1}{2}A^2 v_t \right)\Delta t\right.\nonumber \\
& &\left. -\, \big(A\sqrt{v_t^+}+\gamma_t \rho\big) \big(\rho \sqrt{\Delta t} Z_1 + \sqrt{1-\rho^2}\sqrt{\Delta t} Z_2\big) + \gamma_t \sqrt{\Delta t} Z_1 \right],\\
\label{Xt-discretization}
\bar{X}_{t+\Delta t} &=& \bar{X}_t + r \bar{X}_t \Delta t + \bar{\pi}_t \bar{X}_t \sqrt{v_t^+}\left(A\sqrt{v_t^+} \Delta t + \rho \sqrt{\Delta t} Z_1 + \sqrt{1-\rho^2}\sqrt{\Delta t} Z_2 \right),
\end{eqnarray*}
where $Z_1$ and $Z_2$ are two independent standard normal variables. For wealth process $\bar{X}_t$ driven by $\bar{\pi}_t$, it is possible that an investor loses all his money during the investment period. Thus if $\bar{X}_t\leq 0$, we stop generating the paths, and set $\bar{X}_T=0$ for the current path.

Next we describe the Monte-Carlo methods for computing the tight lower and upper bounds at time $0$. The tight lower and upper bounds at other time $t$ can be computed similarly. Assume $X_0=x$, $v_0=v$ and the dual utility function $\widetilde{U}$ in (\ref{dual-func}) are known. The dual control $\gamma_t=c(t)$ or $c(t)\sqrt{v_t}$ or $c(t)v_t$, where $c$ is a piecewise constant function given by (\ref{piecewise-constant-c}). Denote by $\mathcal{S}$ the set of vectors $\mathbf{C}:=(c_1,\ldots,c_n)$ which form the coefficients of the function $c$.

\bigskip
\noindent {\bf Monte-Carlo method for computing tight lower and upper bounds:}
\begin{itemize}
\item[] Step 1: Fix a vector $\mathbf{C}\in \mathcal{S}$ and a form of dual control $\gamma_t$.
\item[] Step 2: Generate  $M$ sample paths of Brownian motion $W^s$ and $W^v$, discretize SDE (\ref{Y-process}),  compute $Y_T$ with  $Y_0=y$ and  the average derivative:
    $$
    \frac{\partial{\mathcal{Z}(0,y,v)}}{{\partial y}}\approx \frac{1}{y}
    \frac{1}{M}\sum_{\ell=1}^{M} Y_T\widetilde{U}'(Y_T).
    $$
\item[] Step 3: Use the bisection method to solve equation \eqref{nonlinear-algebraic-eq}  and get the solution $y \approx y^*$.
\item[] Step 4:  Compute the upper bound
     $$       \overline{\mathcal{W}}(0,x,v)\approx \mathcal{Z}(0,y^*,v)+xy^*.$$
\item[] Step 5: Find the control process $\bar{\pi}$ in \eqref{pi-expression} and generate the wealth process $\bar{X}$ in (\ref{X-process}).
\item[] Step 6: Compute the lower bound
     $$      \underline{\mathcal{W}}(t,x,v)\approx \frac{1}{M}\sum_{\ell=1}^{M} U(\bar{X}_T).$$
\item[] Step 7:  Repeat Steps 1 to 6 with different $\mathbf{C}\in \mathcal{S}$  to derive the tight lower bound $\sup_{\mathbf{C}\in \mathcal{S}} \underline{\mathcal{W}}(0,x,v)$ and
the tight upper bound  $\inf_{\mathbf{C}\in \mathcal{S}} \overline{\mathcal{W}}(0,x,v)$.
\end{itemize}

\begin{rem} It is  much more time consuming to compute the tight lower bound than to the tight upper bound. The reason is that one has to generate sample paths of the wealth process $\bar X$ and control process $\bar\pi$, which requires to solve equation \eqref{nonlinear-algebraic-eq} at all grid points of time, not just at $t=0$ as in the case of computing the tight upper bound. One technique to speed up is to use a four-dimensional matrix $\bar\pi_{i_t \times j_x \times k_v \times l_C}$ to pre-save the values of $\bar\pi$ on a lattice, and then apply linear interpolation to approximate the exact values we need while generating sample paths of $\bar X$.
\end{rem}

\section{Closed-form upper bounds}\label{sec:examples}

For general dual controls $\gamma_t$, $0\leq t\leq T$, we have to use the Monte Carlo method to compute the upper bound $\mathcal{Z}(t,y,v) $ of  (\ref{new-dual-val-function}). However, for a class of special dual controls and utility functions, we can find the  upper bound in closed-form.
Since $Y$ satisfies a linear SDE \eqref{Y-process} and $\widetilde{U}$ is a decreasing and convex function, $\mathcal{Z}(t,y,v)$ is a decreasing and convex function for $y>0$ with fixed $t$ and $v$. Moreover, the Feynman-Kac theorem implies that $\mathcal{Z}$ satisfies the following linear PDE:
\begin{equation}\label{BS-equation}
  \mathcal{Z}_t - ry\mathcal{Z}_y + \kappa(\theta-v)\mathcal{Z}_v +\frac{1}{2}\mathcal{Z}_{yy}y^2[A^2v +\gamma^2(1-\rho^2)]  + \mathcal{Z}_{yv}y[\gamma\xi\sqrt{v}(1-\rho^2)-Av\xi\rho] +\frac{1}{2}\mathcal{Z}_{vv}\xi^2 v=0
\end{equation}
with terminal condition $\mathcal{Z}(T,y,v)=\widetilde{U}(y)$.
The choice  $\gamma_t=c(t)\sqrt{v_t}$, where $c$ is a piecewise constant function, is particularly interesting as we can get the closed-form solution of the equation (\ref{BS-equation}) if $\widetilde{U}$ is a linear combination of power functions. Specifically, if
$$\gamma_t=c(t)\sqrt{v_t},$$
where  $c$ is a piecewise constant function defined by \eqref{piecewise-constant-c}, and
\begin{equation}\label{sumofUi}
 \widetilde{U}(y) = \sum_{i=1}^K \widetilde{U}_i(y),
\end{equation}
where $\widetilde{U}_i(y)=-(1/q_i)y^{q_i}$ with $q_i<0$   for $i=1,\ldots,K$. The solution of (\ref{BS-equation}) is given by
\begin{equation}\label{sumofsolution}
 \mathcal{Z}(t,y,v) = \sum_{i=1}^K  \widetilde{U}_i(y) \exp(C_i(t)+D_i(t)v),
\end{equation}
where $C_i$ and $D_i$ satisfy the following ODEs
\[
C_i'(t)=d_{1i} D_i(t) + d_{2i}, \quad C_i(T)=0
\]
and
\[
D_i'(t)=a_i D_i^2(t) + b_i(t) D_i(t) + \eta_i(t), \quad D_i(T)=0
\]
with coefficients given by
\begin{eqnarray*}
&&d_{1i}=-\kappa\theta, \quad d_{2i}=rq_i, \\
 &&a_i=-(1/2)\xi^2,\\
&&b_i(t)=\kappa - q_i\xi (c(t)(1-\rho^2)-A\rho),\\
&&\eta_i(t) = -(1/2)q_i(q_i-1)(A^2+ c^2(t)(1-\rho^2)).
\end{eqnarray*}
Furthermore, $D_i$ is given by
$$ D_i(t)=\sum_{j=1}^n D_{ij}(t)1_{(t_{j-1},t_j]}(t),$$
where $D_{ij}$, $j=1,\ldots,n$, are computed recursively as follows: for $j=n$,
$$ D_{in}'(t)=a_i D_{in}^2(t) + b_{i}(t_n) D_{in}(t) + \eta_{i}(t_n), \quad t\in [t_{n-1},t_n] $$
with  terminal condition $D_{in}(t_n)=0$ and, for $j=n-1,\ldots,1$,
$$ D_{ij}'(t)=a_i D_{ij}^2(t) + b_{i}(t_j) D_{ij}(t) + \eta_{i}(t_j), \quad t\in [t_{j-1},t_j] $$
with  terminal condition $D_{ij}(t_j)=D_{i,j+1}(t_j)$. The closed-form solutions of $C_{ij}(t)$ and $D_{ij}(t)$
are given by  (\ref{Ct-2}) and (\ref{Dt-2}) respectively in Appendix A. Comparing $\mathcal{Z}$ in Remark \ref{rk3.2} and  (\ref{sumofsolution}), we see that
$$ F_i(t,v)= \exp(C_i(t)+D_i(t)v)$$
and the upper bound  $ \overline{\mathcal{W}}$ and the feedback control $\bar\pi$ are given by
\begin{equation} \label{uppW}
\overline{\mathcal{W}}(t,x,v)=\sum_{i=1}^K \widetilde U_i(y^*) F_i(t,v) + xy^*
\end{equation}
and
\begin{equation} \label{feedback}
x\bar\pi(t)=\sum_{i=1}^K [A(1-q_i)+\xi\rho D_i(t)](y^*)^{q_i-1}F_i(t,v),
\end{equation}
where $y^*=y(t,x,v)$ is the unique solution of equation $\sum_{i=1}^K y^{q_i-1} F_i(t,v)=x$.

Since the PDE \eqref{BS-equation} can be solved with a closed form solution, which makes the computation of the upper bound very fast. Even if the dual utility is not in the form of (\ref{sumofUi}), but has some simple structure such as call/put option payoff function, one can still compute the upper bound efficiently by using the fast Fourier transform method. We next discuss several examples to illustrate these points.

\begin{exam}\label{power-uppper-bound-Csqrtv}(power utility).
For $U(x)=\frac{x^p}{p}$, its dual function is given by
$\widetilde U(y) = -\frac{y^q}{q}$, where $q=p/(p-1)$. Let $\gamma_t=c\sqrt{v_t}$. This is a special case of \eqref{sumofUi} with $K=1$ and $q_1=q$.
The  dual value function $\mathcal{Z}$, defined by \eqref{new-dual-val-function}, is given by \eqref{sumofsolution}. For power utility,  the upper bound  $ \overline{\mathcal{W}}$ and the feedback control $\bar\pi$ can be written out explicitly as
$$  \overline{\mathcal{W}}(t,x,v) = U(x)\exp((1-p)(C(t)+D(t)v))
\mbox{ and }\bar{\pi}(t) = (1-q)A + \xi\rho D(t),$$
where $C(t)$ and $D(t)$ are given by \eqref{Ct-2} and \eqref{Dt-2},  respectively,  with $\underline t=0$,
$\bar t=T$ and $f_1=f_2=0$.
Note that $\bar{\pi}$ is a deterministic function of time $t$.
We can then use the Monte Carlo method to generate sample paths of the wealth process to compute the lower bound, see Remark \ref{rk3.2}.  However, for power utility, there is a fast  approximation method to compute the lower bound as shown next.
By the Feynman-Kac theorem, the lower bound $\underline{W}$, defined by \eqref{Wlowerbound-C},  satisfies the following PDE:
\begin{eqnarray*}\label{power-lower-bound-2}
\frac{\partial \underline{\mathcal{W}}}{\partial t} + (r + A\bar{\pi}(t)v) x  \underline{\mathcal{W}}_x  + \kappa(\theta-v)\underline{\mathcal{W}}_v + \frac{1}{2}\bar{\pi}(t)^2x^2v \underline{\mathcal{W}}_{xx}+\, \xi \bar{\pi}(t) x v \rho  \underline{\mathcal{W}}_{xv} +\frac{1}{2}\xi^2 v \underline{\mathcal{W}}_{vv}=0
\end{eqnarray*}
with the terminal condition $\underline{\mathcal{W}}(T,x,v)=\frac{x^p}{p}$. Thanks to the power utility, the solution  of the above equation is given by
$$\underline{\mathcal{W}}(t,x,v) = \frac{x^p}{p}\exp\left[\bar{C}(t)+\bar{D}(t)v\right],$$
where $\bar{C}(t)$ and $\bar{D}(t)$ satisfy the following  ODEs
\begin{equation*}\label{ODE-Ct-2}
\bar{C}'(t)=-\kappa\theta \bar{D}(t) - rp,\quad \bar{C}(T)=0
\end{equation*}
and
\begin{equation*}\label{ODE-Dt-2}
\bar{D}'(t)=-\frac{1}{2}\xi^2\bar{D}^2(t) -(\xi \bar{\pi}(t) \rho p -\kappa) \bar{D}(t) - Ap\bar{\pi}(t) - \frac{1}{2}\bar{\pi}(t)^2p(p-1),\quad
\bar{D}(T)=0.
\end{equation*}
 Even though $\bar D$ satisfies a Riccati equation, there is no closed form solution for $\bar D$ as $\bar \pi$ is a continuous function, not a constant. We can nevertheless approximate $\bar\pi$ with a piecewise constant function and then get a closed-form approximate solution to $\bar D$ with a recursive method. Specifically, we may divide interval $[0,T]$ by grid points $0=\tilde t_0<\tilde t_1<\ldots \tilde t_m=T$ and approximate $\bar\pi$ by a piecewise constant function
$$ \tilde\pi(t)=\sum_{k=1}^m \bar\pi(\tilde t_k)1_{(t_{k-1},t_k]}(t).$$
$\tilde\pi$ can be made arbitrarily close to $\bar\pi$. If we replace $\bar\pi$ by $\tilde\pi$ in the Riccati equation for $\bar D$, the solution of the resulting equation can be written as
$$ \tilde D(t) = \sum_{k=1}^m \tilde D_k(t) 1_{(t_{k-1},t_k]}(t),$$
where $\tilde D_k$, $k=m,\ldots,1$, satisfy Riccati equations \eqref{Dt-2} with constant coefficients on intervals $(\tilde t_{k-1},\tilde t_k]$  and  can be computed recursively in a closed form
 with terminal conditions $\tilde D_k(\tilde t_k)=\tilde D_{k+1}(\tilde t_k)$. The function $\tilde D$ is a good approximation of $\bar D$.
\end{exam}

\begin{exam}\label{non-HARA-uppper-bound-Csqrtv} (non-HARA utility).
Assume
\begin{equation*}\label{non-HARA-utility}
U(x)={1\over 3}H(x)^{-3}+ H(x)^{-1} + xH(x),
\end{equation*}
for $x>0$, where
$$ H(x)=\left({2\over -1+\sqrt{1+4x}}\right)^{1/2}.$$
It can be easily checked that
$U$ is continuously differentiable, strictly increasing and strictly concave, satisfying  $U(0)=0$,  $U(\infty)=\infty$, $U'(0)=\infty$ and $U'(\infty)=0$.   Furthermore, the relative risk aversion coefficient of $U$ is given by
$$ R(x)=-{xU''(x)\over U'(x)}= {1\over 4}\left(1+{1\over \sqrt{1+4x}}\right),$$
which shows that $U$ is not a HARA utility  and represents an investor who will increase the percentage of wealth invested in the risky asset as wealth increases, see  \cite{Bian2015}  for more details.
The dual function of $U$ is given by
$$\widetilde U(y)={1\over 3}y^{-3}+y^{-1}.$$
Let $\gamma_t=c\sqrt{v_t}$. This is a special case of \eqref{sumofUi} with $K=2$ and $q_1=-3,\, q_2=-1$. The dual value function $\mathcal{Z}$ is given by \eqref{sumofsolution},
the upper bound $\overline{\mathcal{W}}$ by (\ref{uppW}) and the feedback control
$\bar\pi$ by (\ref{feedback}), in the case here, $y^*$ can be computed explicitly as
$$ y^* = \left(\frac{F_2(t,v) + \sqrt{F_2(t,v)^2+4x F_1(t,v)}}{2x}\right)^{\frac{1}{2}},$$
where
$C_i(t),\, D_i(t)$ are given by \eqref{Ct-2} and \eqref{Dt-2} respectively with $\underline t=0$,
$\bar t=T$ and $f_1=f_2=0$, see Appendix B.

Note that, unlike the case for power utility, there is no closed form formula for the lower bound $\underline{W}$. One has to use the Monte Carlo method to generate sample paths of the wealth process in order to find its value. We can nevertheless find a reasonable lower bound at more expensive computational cost.
\end{exam}

\begin{exam} \label{Yaari-upper-bound-Csqrtv} (Yarri utility).
Assume
\begin{equation}\label{threshold-utility}
U(x)=x \wedge L,
\end{equation}
where $L$ is a positive constant. $U$ is a continuous, increasing and concave function, but not differentiable at $x=L$ and not strictly concave. Also note that $U\rq{}(0)=1$, so Inada\rq{}s condition is not satisfied.  This utility is called Yarri utility and is used in behavioural finance.
 The dual function  is given by
\begin{equation*}\label{dual-function-threshold}
\widetilde{U}(y)=L(1-y)^+.
\end{equation*}
For the dual process \eqref{Y-process} with $\gamma_t=c\sqrt{v_t}$, where $c>0$ is an arbitrarily fixed constant, we evaluate the dual value function
\[
\mathcal{Z}(t,y,v) = E_{t,y,v}[\widetilde{U}(Y_T)]=E_{t,y,v}[L(1-Y_{T})^+].
\]
This is a European put option pricing problem with the Heston model.
Let $Z_{T}=\ln Y_{T}$ and $z=\ln y$. Then
\begin{equation}\label{dual-val-func-Yaari}
\mathcal{Z}(t,y,v) = \widetilde{\mathcal{Z}}(t,z,v)=E_{t,z,v}[L(1-e^{Z_{T}})^+],
\end{equation}
with terminal condition $\widetilde{\mathcal{Z}}(T,z,v)= L(1-e^{z})^+$.
Although the conditional probability density function of $Z_{T}$ is unknown, its conditional characteristic function (namely, the Fourier transform of the density function) can be derived. Therefore, analogous to the well-known Heston method in \cite{Heston1993}, function $\mathcal{Z}$ in \eqref{dual-val-func-Yaari} can be written as an integral formula, which can be evaluated by numerical integration rules based on the Fast Fourier Transform (FFT). \cite{Fang2008} develop Fourier-cosine expansion in the context of numerical integration as a more efficient alternative for the methods based on the FFT, which is named as COS method. For the convenience of the readers, we show the main ideas of COS method in Appendix B.

 We now give some details.
Define the conditional characteristic function of $Z_T$ by
\begin{equation*}
\phi(t,z,v;\omega) = E\left[e^{i\omega Z_T}|Z_t=z,\,v_t=v\right].
\end{equation*}
By the Feynman-Kac theorem, $\phi$ satisfies the following PDE
\begin{eqnarray}\label{Yaari-f-PDE}
&& \frac{\partial \phi}{\partial t} - \left\{ r+\frac{1}{2}v[A^2+c^2(1-\rho^2)] \right\}\phi_z + \kappa(\theta-v)\phi_v + \frac{1}{2}v[A^2+c^2(1-\rho^2)]\phi_{zz}\nonumber\\
&&\quad +\,v\xi[c(1-\rho^2)-A\rho]\phi_{zv} + \frac{1}{2}\xi^2 v \phi_{vv}=0.
\end{eqnarray}
Assume that $\phi$ takes the following form
\begin{equation}\label{Yaari-f}
\phi(t,z,v;\omega) = \exp(C(t;\omega) + D(t;\omega)v + i\omega z).
\end{equation}
with $C(T;\omega)=0$ and $D(T;\omega)=0$.
Inserting \eqref{Yaari-f} into \eqref{Yaari-f-PDE} gives that $C$ and $D$ satisfy the Riccati equations in Appendix A with coefficients
$d_1=-\kappa\theta,\, d_2=r i \omega$ and
$$ a= -\frac{1}{2}\xi^2, \quad
b= - \{ \xi[c(1-\rho^2)-A\rho]i\omega-\kappa \}, \quad
\eta=\frac{1}{2}[A^2 + c^2(1-\rho^2)](\omega^2 + i\omega).
$$
 The closed form solutions $C$ and $D$  are given by \eqref{Ct-2} and \eqref{Dt-2}, respectively, with $\underline{t}=0$, $\bar{t}=T$, $f_1=0$, $f_2=0$.
Define $\varphi(t,v;\omega)=e^{-i\omega z}\phi(t,z,v;\omega)$. This is the conditional characteristic function of $Z_T-Z_t= \ln Y_T-\ln Y_{t}$.

Following  \cite{Fang2008}, we can easily find that the upper bound is given by
\begin{equation}\label{Yaari-upper-bound-2}
 \overline{\mathcal{W}}(t,x,v) = \mathcal{Z}(t,y^*,v) + xy^*,
\end{equation}
where
\begin{equation}\label{Yaari-dual-val-function}
\mathcal{Z}(t,y^*,v):=\widetilde{\mathcal{Z}}(t,z^*,v) \approx
\mathop{{\sum}'} \limits_{k=0}^{N-1} \mathrm{Re} \left\{ \varphi\left(t,v;\frac{k\pi}{\zeta_{2}-\zeta_{1}}\right)e^{ik\pi\frac{z^*-\zeta_{1}}{\zeta_{2}-\zeta_{1}}} \right\}\widetilde{\mathcal{Z}}_k,
\end{equation}
and $\widetilde{\mathcal{Z}}_k,\, y^*$ and other constants are given in Appendix B.
The feedback control for computing the lower bound is given by \eqref{pi-expression}.
\end{exam}

\section{Numerical tests}\label{sec:numerical-test}

In the following numerical examples we use the dual-control Monte-Carlo method to solve the optimal control problem \eqref{optimization-W} with power, non-HARA and Yarri utilities. We compute the upper bounds using the closed form formulas for power and non-HARA utilities and the Fourier-cosine method for Yarri utility when $\gamma=c\sqrt{v}$ and everything else (the lower bounds for all $\gamma$ and the upper bounds for $\gamma=c$ and $\gamma=cv$) using the Monte-Carlo method with path number 100,000 and time steps 100 for discretizing SDEs with the Euler method, see Remarks \ref{rk3.2} and \ref{rk3.3}.


\subsection{Power utility}

\begin{exam}\label{exam:power}
This example is aimed to apply the lower and upper bound method to the power utility when $v_t$ following mean-reversion square-root process. The following parameters
\begin{equation}\label{parameters}
r = 0.05, ~\rho = -0.5, ~\kappa = 10, ~\theta = 0.05, ~\xi = 0.5, ~A = 0.5, ~x_0=1, ~v_0 = 0.5, ~T = 1,
\end{equation}
are taken from \cite{Zhang2016}.
The comparisons are carried out for the cases of sampling control $c$ for $1,5,20,80$ times uniformly distributed in $[-0.5,0.5]$ for both the lower and upper bounds. The benchmark value is the primal value explicitly given by \cite{Kraft2005}. The parameter $p$ in utility function equals $1/2$, and other parameters follow values in \eqref{parameters}. The numerical results are listed in Table~\ref{power-table-1}.

\begin{table}[!htbp]
\caption{Lower bound (LB) and upper bound (UB) for power utility (Example~\ref{exam:power}). The benchmark from Example \ref{power-primal-analytical} equals 2.074842.  } 
\centering 
\begin{tabular}{l c c c c c r}\label{power-table-1}\\
\hline
\multicolumn{7}{c}{$\gamma_t = c$}\\
Num $c$    &LB	            &UB	            &diff	              &rel-diff (\%)	        &LB time (secs) &UB time (secs)\\
\hline
1	       &2.074824	    &2.074894	&7.01e$-5$ & 3.38\text{e}$-3$	&2.62\text{e}$+4$	   &$2.90\text{e}$$+0$\\
5	       &2.074824	    &2.074894	&7.01\text{e}$-5$ & 3.38\text{e}$-3$	&2.63\text{e}$+4$	   &$1.46\text{e}$$+1$\\
20	       &2.074824	    &2.074894	&7.01\text{e}$-5$ & 3.38\text{e}$-3$	&2.65\text{e}$+4$	   &$5.59\text{e}$$+1$\\
80	       &2.074824	    &2.074893	&6.92\text{e}$-5$ & 3.34\text{e}$-3$	&2.76\text{e}$+4$	   &2.24\text{e}$+2$\\
\hline
\multicolumn{7}{c}{$\gamma_t = c\sqrt{v_t}$}\\
Num $c$    &LB	            &UB	            &diff	              &rel-diff (\%)	        &LB time (secs) &UB time (secs)\\
\hline
1	       &2.074823	    &2.074845	&2.12\text{e}$-5$ &$1.02\text{e}$$-3$	&$4.65\text{e}$$+0$    &$9.87\text{e}$$-4$\\
5	       &2.074823	    &2.074845	&$2.12\text{e}$$-5$ &$1.02\text{e}$$-3$	&$2.33\text{e}$$+1$    &$3.30\text{e}$$-3$\\
20	       &2.074823	    &2.074845	&$2.12\text{e}$$-5$ &$1.02\text{e}$$-3$	&$9.57\text{e}$$+1$    &$5.68\text{e}$$-3$\\
80	       &2.074823	    &2.074842	&$1.88\text{e}$$-5$ &$9.06\text{e}$$-4$	&$3.91\text{e}$$+2$    &$8.18\text{e}$$-3$\\
\hline
\multicolumn{7}{c}{$\gamma_t = c v_t$}\\
Num $c$    &LB	            &UB	            &diff	              &rel-diff (\%)	        &LB time (secs) &UB time (secs)\\
\hline
1	       &2.074824	    &2.074894	&$7.01\text{e}$$-5$ &$3.38\text{e}$$-3$	&$2.64\text{e}$$+4$	   &$2.83\text{e}$$+0$\\
5	       &2.074824	    &2.074894	&$6.97\text{e}$$-5$ &$3.36\text{e}$$-3$	&$2.65\text{e}$$+4$	   &$1.38\text{e}$$+1$\\
20	       &2.074824	    &2.074894	&$6.97\text{e}$$-5$ &$3.36\text{e}$$-3$	&$2.67\text{e}$$+4$	   &$5.37\text{e}$$+1$\\
80	       &2.074824	    &2.074893	&$6.88\text{e}$$-5$ &$3.31\text{e}$$-3$	&$2.79\text{e}$$+4$	   &$2.15\text{e}$$+2$\\
\hline
\end{tabular}
\end{table}
\end{exam}

\begin{exam}\label{exam:power-robust-test}
From Example \ref{exam:power}, we see $\gamma_t = c\sqrt{v_t}$ outperforms other choices of $\gamma_t$. In this example we further test the robustness of the dual control Monte-Carlo methods for $\gamma_t = c\sqrt{v_t}$. The comparisons are carried out for the cases of sampling control $c$ for $1,5,20,80$ times uniformly distributed in $[-0.5,0.5]$ both for the lower and upper bounds. In Table \ref{power-table-2}, we give the mean and standard deviation of the absolute and relative difference between the lower and upper bounds of power utility with randomly sampled parameters-sets: 10 samples of $r$ from the uniform distribution on interval $[0.01,0.08]$, $\rho$ on $[-1,1]$, $\kappa$ on $[1,10]$, $\theta$ on $[0.01,1]$, $\xi$ on $[0.1,1]$, $A$ on $[0.1,1.5]$, $x_0=1$, $v_0=0.5$, and $T=1$. It is clear that the gap between the tight lower and upper bounds is very small, especially when the dual control $c$ is used. This shows that the algorithm is reliable and accurate. The numerical results are listed in Table~\ref{power-table-2}.

\begin{table}[!htbp]
\caption{Mean and std of the absolute and relative difference between the lower and upper bounds for power utility (Example~\ref{exam:power-robust-test}) with many randomly sampled parameters-sets.} 
\centering 
\begin{tabular}{l c c c c r}\label{power-table-2}\\
\hline
Num $c$  &mean diff               &std diff               &mean rel-diff (\%)      &std rel-diff (\%)       &mean time (secs) \\
1	   &$2.2695\text{e}$$-3$	&$2.7658\text{e}$$-3$	&$9.8159\text{e}$$-2$	&$1.1543\text{e}$$-1$	&$4.29\text{e}$$+1$\\
5	   &$2.2660\text{e}$$-3$	&$2.7632\text{e}$$-3$	&$9.8010\text{e}$$-2$	&$1.1532\text{e}$$-1$	&$2.16\text{e}$$+2$\\
20	   &$2.0003\text{e}$$-3$	&$2.3089\text{e}$$-3$	&$8.6739\text{e}$$-2$	&$9.6877\text{e}$$-2$	&$8.72\text{e}$$+2$\\
80	   &$1.8253\text{e}$$-3$	&$1.9986\text{e}$$-3$	&$7.9391\text{e}$$-2$	&$8.4384\text{e}$$-2$	&$3.48\text{e}$$+3$\\
\hline
\end{tabular}
\end{table}
\end{exam}

\begin{exam}\label{exam:power-piecewise-constant}
 This example compares performances of $\gamma_t = c(t)\sqrt{v_t}$ with $c(t)$ being a constant ($c(t)=c$, the number of pieces $n_1=1$) and being a two-piecewise constant function
($c(t)=c_1 1_{[0,T/2]}(t) + c_2 1_{(T/2,T]}(t)$,  the number of pieces $n_1=2$).
Since $U$ is a power utility, we also replace the feasible control $\bar\pi$ for the lower bound by a piecewise constant control ($\tilde \pi(t)=\sum_{k=1}^{n_2} \bar\pi(\tilde t_k)1_{(\tilde t_{k-1}, \tilde t_k]}(t)$ with $\tilde t_k=k(T/n_2)$ and $n_2=100$) to expedite the computation of the lower bound.  Table~\ref{power-table-3} lists the numerical results. It is clear that lower and upper bounds are very tight, even for $n_1=1$ and one sample of constant $c$ which is 0 in this case. The lower bound is the same as the optimal value. The upper bound can be improved as the number of samples for $c$ is increased. We make the number of samples for each  $c_i$, $i=1,2$, the same as that of $c$ for $n_1=1$ to ensure piecewise functions with $n_1=2$ include all functions with $n_1=1$, so the performance should be better. Our numerical results confirm this is indeed the case, even though the rate of improvement is small, possibly because the bounds are already very tight. This implies we can reduce the gap of the bounds by increasing the number $n_1$ at cost of exponentially  increased computation. One needs to strike a balance of accuracy and cost. Since $n_1=1$ gives good estimation of the bounds, we use it from now on for other utilities too, including non-HARA and Yarri utilities.

\begin{table}[!htbp]
\caption{Lower bound (LB) and upper bound (UB) for power utility with piecewise constant control $c$ (Example~\ref{exam:power-piecewise-constant}). The benchmark from Example \ref{power-primal-analytical} equals 2.074842060.  } 
\centering 
\begin{tabular}{l c c c c c r}\label{power-table-3}\\
\hline
\multicolumn{7}{c}{$n_1=1$}\\
Num $c$    &LB	            &UB	            &diff	              &rel-diff (\%)	        &LB (secs) &UB  (secs)\\
\hline
1	          &2.074842060 	    &2.074844628	&$2.5680\text{e}$$-6$ &$1.2377\text{e}$$-4$	    &$1.75\text{e}$$-4$	&$1.74\text{e}$$-4$\\
60	          &2.074842060 	    &2.074844628	&$2.5680\text{e}$$-6$ &$1.2377\text{e}$$-4$	    &$1.00\text{e}$$-3$	&$8.72\text{e}$$-4$\\
600	          &2.074842060 	    &2.074842126	&$6.5705\text{e}$$-8$ &$3.1667\text{e}$$-6$	    &$1.30\text{e}$$-2$	&$1.80\text{e}$$-3$\\
6000	      &2.074842060 	    &2.074842125	&$6.4925\text{e}$$-8$ &$3.1292\text{e}$$-6$	    &$1.11\text{e}$$-1$	&$4.02\text{e}$$-3$\\
\hline
\multicolumn{7}{c}{$n_1=2$}\\
Num $c$    &LB	            &UB	            &diff	              &rel-diff (\%)	        &LB (secs) &UB (secs)\\
\hline
1	          &2.074842060 	    &2.074844628	&$2.5680\text{e}$$-6$ &$1.2377\text{e}$$-4$  	&$2.87\text{e}$$-4$	&$1.72\text{e}$$-4$\\
$60^2$ 	      &2.074842060 	    &2.074842469	&$4.0910\text{e}$$-7$ &$1.9717\text{e}$$-5$	    &$9.10\text{e}$$-2$	&$2.05\text{e}$$-3$\\
$600^2$  &2.074842060 	    &2.074842119	&$5.9149\text{e}$$-8$ &$2.8508\text{e}$$-6$	    &$7.75\text{e}$$+0$	&$1.24\text{e}$$-1$\\
$6000^2$  &2.074842060 	    &2.074842104	&$4.4493\text{e}$$-8$ &$2.1444\text{e}$$-6$	    &$7.74\text{e}$$+2$	&$2.84\text{e}$$-1$\\
\hline
\end{tabular}
\end{table}
\end{exam}

\subsection{Non-HARA utility}

\begin{exam}\label{exam:non-HARA-vt-constant}
This example is aimed to check the correctness of the lower and upper bounds when process $v_t$ always constant through the time, in which case there is explicit solution to the primal value function.  Let $v_0 = \theta,\, \xi = 0$, and the other parameters be the same as \eqref{parameters}.
Denote $\bar{W}_1 = \exp[(3 r + 6 A^2\theta)(T-t)]$ and $\bar{W}_2 = \exp[(r+A^2\theta)(T-t)]$. Then the primal value function has the following explicit form (see \cite{Bian2015}):
\begin{equation*}\label{nonHARA-vt-constant-benchmark}
\mathcal{W}(t,x) = \frac{2}{3}\left(\frac{\bar{W}_2}{y^*} + 2 x y^*\right),
\end{equation*}
with
\begin{equation*}
y^* = \sqrt{\frac{1}{2x}\left(\bar{W}_2 + \sqrt{\bar{W}^2_2 + 4 x \bar{W}_1}\right)}.
\end{equation*}
The lower and upper bounds are computed by the Monte-Carlo method with path number $100,000$ and time steps $100$.
The numerical results are  listed in Table~\ref{non-HARA-table-1}, in which the numerics show that the benchmark is between the lower and upper bound, and the difference between these is proportional to $10^{-4}$ and relative difference $10^{-5}$. Therefore, the lower and upper bound methods are reliable and accurate.

\begin{table}[!htbp]
\caption{Lower bound (LB) and upper bound (UB) for Example~\ref{exam:non-HARA-vt-constant} (non HARA utility).} 
\centering 
\begin{tabular}{l c c c r }\label{non-HARA-table-1}\\
\hline
Benchmark   &LB	        &UB	        &diff	    &rel-diff (\%)	\\
\hline
2.307810 	&2.307691 	&2.307843   &$1.52\text{e}$$-4$	&$6.60\text{e}$$-3$\\
\hline
\end{tabular}
\end{table}
\end{exam}

\begin{exam}\label{exam:non-HARA-vt-not-constant}
This example is aimed to apply the lower and upper bound methods to the non HARA utility when $v_t$ following mean-reversion square-root process. The comparisons are carried out for the cases of sampling control $c$ for $20$ times uniformly distributed in $[-0.5,0.5]$ both for the lower and upper bounds. The other parameters values are the same as in \eqref{parameters}.  The numerical results in Table~\ref{non-HARA-table-2} show that the choice $\gamma=c\sqrt{v_t}$ outperforms the others.

\begin{table}[!htbp]
\caption{Lower bound (LB) and upper bound (UB) for non HARA utility (Example~\ref{exam:non-HARA-vt-not-constant}).  } 
\centering 
\begin{tabular}{l c c c c c r}\label{non-HARA-table-2}\\
\hline
$\gamma$    &LB	            &UB	            &diff	                &rel-diff(\%)	        &LB time (secs) &UB time (secs)\\
\hline
$c$	       &2.327407	    &2.327834	&$4.27\text{e}$$-4$	&$1.83\text{e}$$-2$	&$2.43\text{e}$$+4$	   &$5.37\text{e}$$+1$\\
$c\sqrt{v_t}$	       &2.327573	    &2.327858	&$2.84\text{e}$$-4$	&$1.22\text{e}$$-2$	&$1.36\text{e}$$+2$    &$4.35\text{e}$$-3$\\
$c v_t$      &2.327411	    &2.327833	&$4.21\text{e}$$-4$	&$1.81\text{e}$$-2$	&$2.41\text{e}$$+4$	   &$5.58\text{e}$$+1$\\
\hline
\end{tabular}
\end{table}

Using the optimal control $c^*$ for computing the tight lower bound for $\gamma_t = c \sqrt{v_t}$ in Table~\ref{non-HARA-table-2}, we draw the 3D figures for the optimal strategies $\bar{\pi}(t,x,v)$ and the distribution of the terminal wealth (see Figure~\ref{non-HARA-pi-Csqrtv}).
\end{exam}

\begin{figure}[!htbp]
\centering
\subfigure
{
 \begin{minipage}{1.85in}
  \centering
  \includegraphics[width=1.85in,height=1.85in]{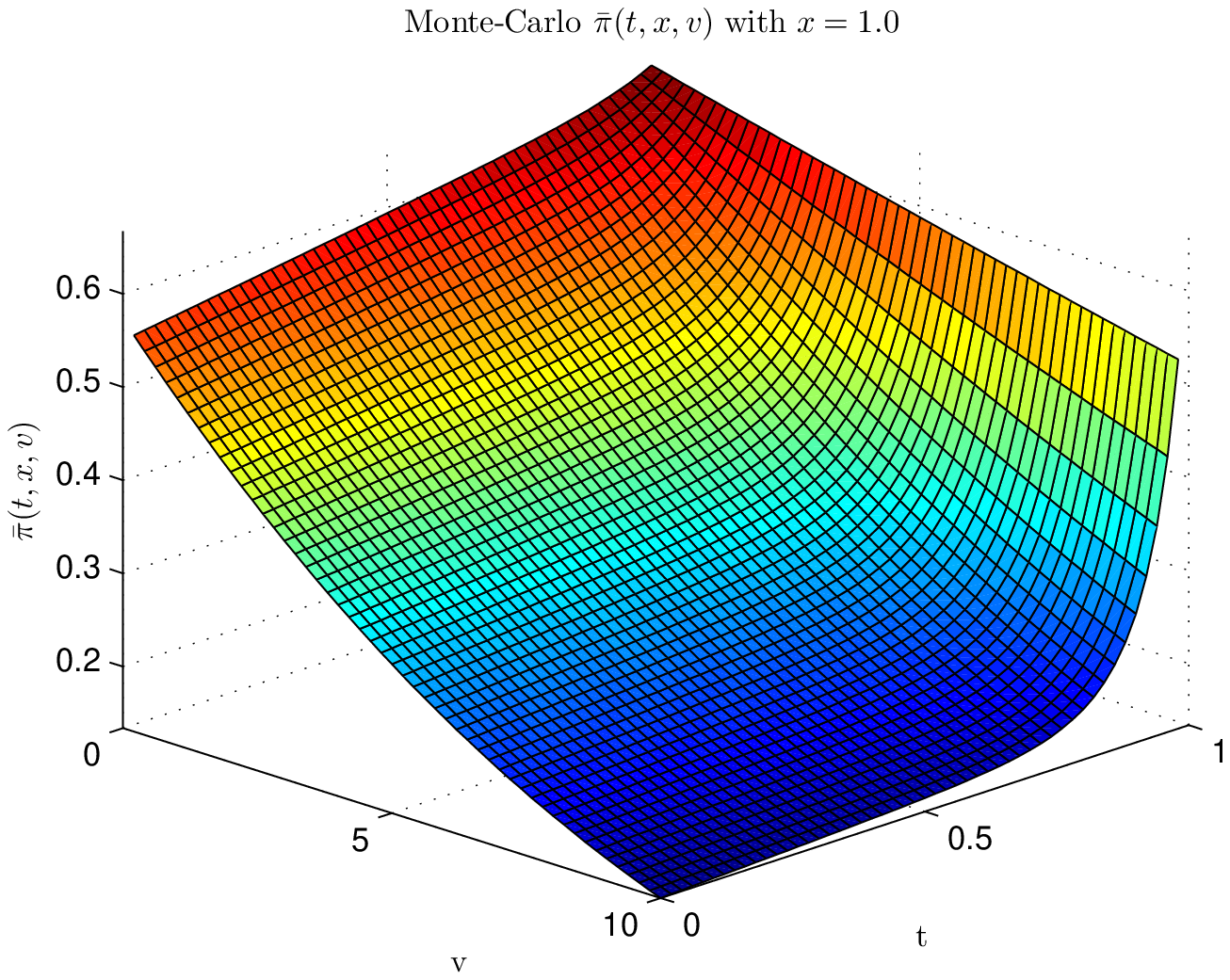}
 \end{minipage}
}
\subfigure
{
 \begin{minipage}{1.85in}
 \centering
 \includegraphics[width=1.85in,height=1.85in]{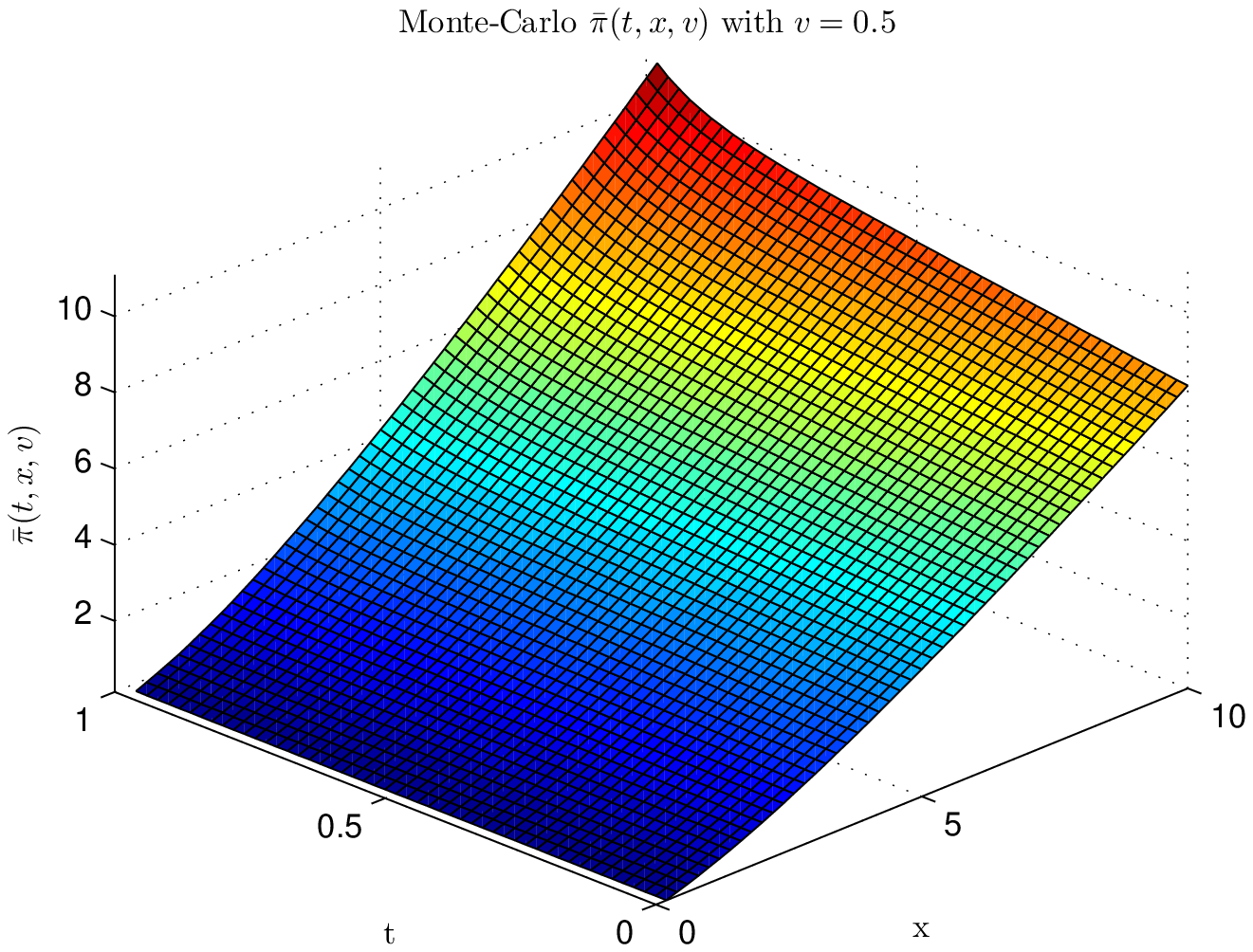}
\end{minipage}
}
\subfigure
{
 \begin{minipage}{1.85in}
 \centering
 \includegraphics[width=1.85in,height=1.85in]{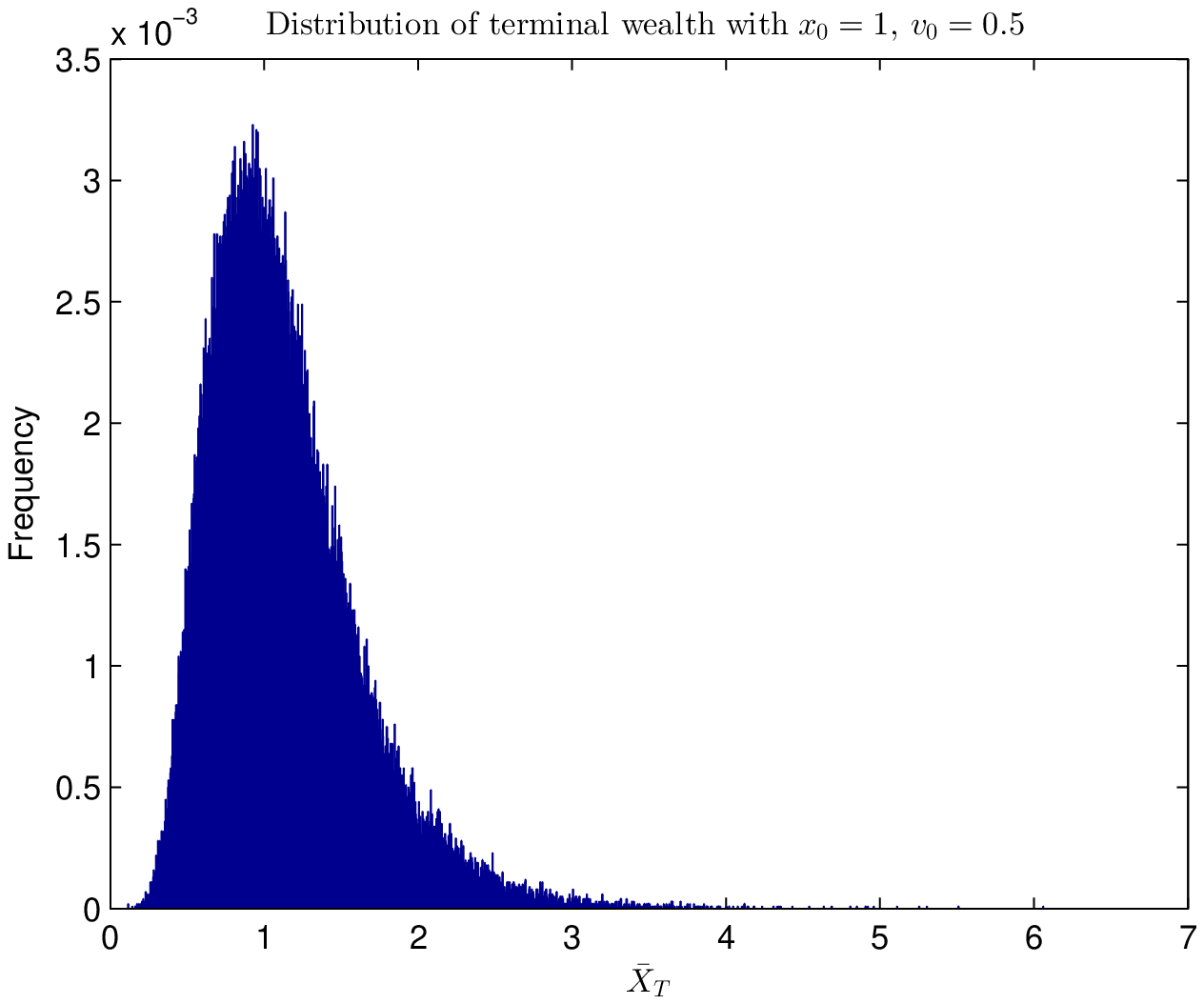}
\end{minipage}
}
\caption{3D and 2D figures for non HARA utility (Example \ref{exam:non-HARA-vt-not-constant}). The left figure is the optimal strategies $\bar{\pi}(t,x,v)$ with initial wealth $x_0=1$. The middle figure is the optimal strategies $\bar{\pi}(t,x,v)$ with initial variance $v_0 = 0.5$. And the right figure is the distribution of terminal wealth process.} %
\label{non-HARA-pi-Csqrtv}
\end{figure}

\begin{exam}\label{exam:non-HARA-robust-test}
In this example, we further examine the robustness of the lower and upper bound methods with $\gamma_t = c\sqrt{v_t}$. The comparisons are carried out for the cases of sampling control $c$ for $1,5,10,20$ times uniformly distributed in $[-0.5,0.5]$ both for the lower and upper bounds. In Table \ref{non-HARA-table-3}, we give the mean and standard deviation of the absolute and relative difference between the lower and upper bounds for non-HARA utility with randomly sampled parameters-sets: $10$ samples of $r$ from the uniform distribution on interval $[0.01,0.08]$, $\rho$ on $[-1,1]$, $\kappa$ on $[1,10]$, $\theta$ on $[0.01,1]$, $\xi$ on $[0.1,1]$, $A$ on $[0.1,1.5]$, $x_0=1$, $v_0=0.5$, and $T=1$. It is clear that the difference (and relative difference) between the tight lower and upper bounds is very small, especially when the dual control $c$ is used. This shows that the algorithm is reliable and accurate.

\begin{table}[!htbp]
\caption{Mean and std of the absolute and relative difference between the lower and upper bounds for non-HARA utility (Example~\ref{exam:non-HARA-robust-test}) with many randomly sampled parameters-sets.} 
\centering 
\begin{tabular}{l c c c c r}\label{non-HARA-table-3}\\
\hline
Num $c$  &mean diff               &std diff               &mean rel-diff (\%)      &std rel-diff (\%)       &mean time (secs) \\
1	   &$1.1434\text{e}$$-2$	&$1.7040\text{e}$$-2$	&$4.2995\text{e}$$-1$	&$6.3311\text{e}$$-1$	&$1.33\text{e}$$+2$\\
5	   &$9.3788\text{e}$$-3$	&$1.2686\text{e}$$-2$	&$3.5362\text{e}$$-1$	&$4.6995\text{e}$$-1$	&$6.64\text{e}$$+2$\\
10	   &$8.7469\text{e}$$-3$	&$1.2598\text{e}$$-2$	&$3.2842\text{e}$$-1$	&$4.6679\text{e}$$-1$	&$1.33\text{e}$$+3$\\
20	   &$8.5262\text{e}$$-3$	&$1.2124\text{e}$$-2$	&$3.1998\text{e}$$-1$	&$4.4861\text{e}$$-1$	&$2.66\text{e}$$+3$\\
\hline
\end{tabular}
\end{table}
\end{exam}

\subsection{Yarri utility}

\begin{exam}\label{exam:Yarri-vt-constant}
This example is aimed to check the lower and upper bound methods when process $v_t$ always constant through the time. Let $v_0 = \theta,\, \xi = 0$, and the other parameters be the same as \eqref{parameters}.
Then the analytical solution to the primal value is given by
\begin{equation*}\label{Yarri-vt-constant-benchmark}
\mathcal{W}(t,x)=
\left\{
\begin{array}{ll}
L \Phi\{ \Phi^{-1}[\frac{x}{L}e^{r(T-t)}] + A\sqrt{\theta (T-t)}\}, &0\leq x < L e^{-r(T-t)},\\
L, &x\geq L e^{-r(T-t)}.
\end{array}\right.
\end{equation*}
The upper bound is computed by the Monte-Carlo methods with path number $10,000$ and time steps $100$, and  the lower bound with path number $100,000$ and time steps $100$. The threshold $L$ is taken as $L=2$. The numerical results are listed in Table~\ref{Yarri-table-1}, which confirm that the lower and upper bound methods are reliable and accurate.

\begin{table}[h]
\caption{Lower bound (LB) and upper bound (UB) for Yarri utility (Example~\ref{exam:Yarri-vt-constant}).} 
\centering 
\begin{tabular}{l c c c r }\label{Yarri-table-1}\\
\hline
Benchmark   &LB	        &UB	        &diff	    &rel-diff (\%)	\\
\hline
1.139790	&1.136091 	&1.139889   &$3.80\text{e}$$-3$	&$3.34\text{e}$$-1$\\
\hline
\end{tabular}
\end{table}
\end{exam}

\begin{exam}\label{exam:Yarri-vt-not-constant}
This example is aimed to apply the lower and upper bound methods to the Yarri utility when $v_t$ following mean-reversion square-root process.
The comparisons are carried out for sampling control $c$ for $20$ times uniformly distributed in $[-0.5,0.5]$ both for the lower and upper bounds. The values of other parameters are the same as Example \ref{exam:Yarri-vt-constant}.
 For the Fourier-cosine methods, we set the truncation number as $N=64$.  The numerical results are listed in Table~\ref{Yarri-table-2}. It is shown that the choice $\gamma_t = c\sqrt{v_t}$ outperforms the others.

\begin{table}[h]
\caption{Lower bound (LB) and upper bound (UB) for Yarri utility (Example~\ref{exam:Yarri-vt-not-constant}).} 
\centering 
\begin{tabular}{l c c c c c r}\label{Yarri-table-2}\\
\hline
   $\gamma$ &LB	            &UB	            &diff	                &rel-diff (\%)	        &LB time (secs) &UB time (secs)\\
\hline
$c$       &1.113889	    &1.174928	&$6.10\text{e}$$-2$	&$5.48\text{e}$$+0$	&$1.08\text{e}$$+4$ &$6.08\text{e}$$+0$\\
$c\sqrt{v_t}$       &1.172057	    &1.173366	&$1.31\text{e}$$-3$	&$1.12\text{e}$$-1$	&$2.37\text{e}$$+3$ &$2.48\text{e}$$-1$\\
$c v_t$	       &1.137594	    &1.174928	&$3.73\text{e}$$-2$	&$3.28\text{e}$$+0$	&$1.09\text{e}$$+4$	&$5.69\text{e}$$+0$\\
\hline
\end{tabular}
\end{table}

The 3D figures are drawn for the optimal strategy $\bar{\pi}(t,x,v)$. Also it is plotted that the distribution of the terminal wealth (See Figure~\ref{Yarri-pi-Csqrtv}).
\end{exam}

\begin{figure}[!htbp]
\centering
\subfigure
{
 \begin{minipage}{1.85in}
  \centering
  \includegraphics[width=1.85in,height=1.85in]{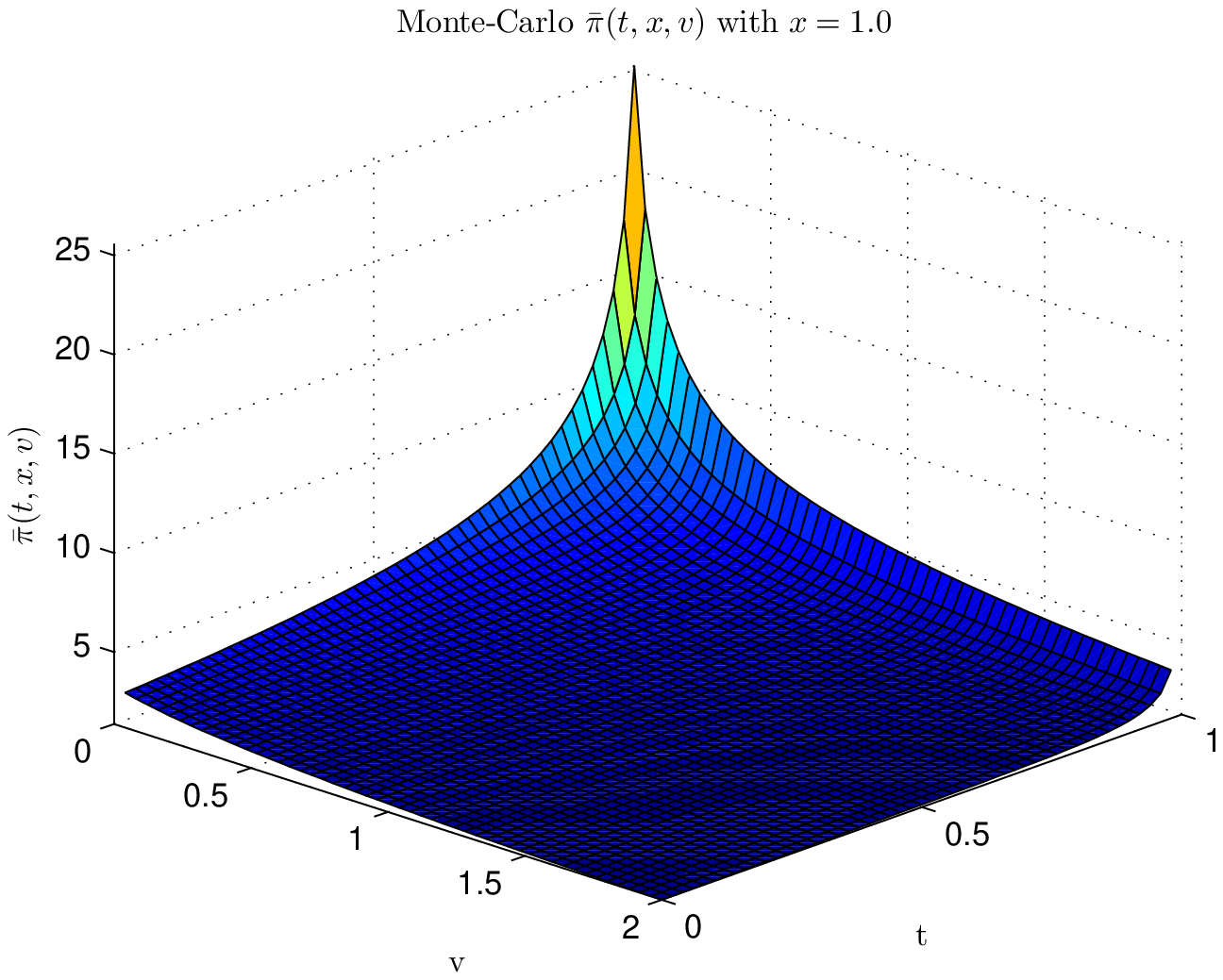}
 \end{minipage}
}
\subfigure
{
 \begin{minipage}{1.85in}
 \centering
 \includegraphics[width=1.85in,height=1.85in]{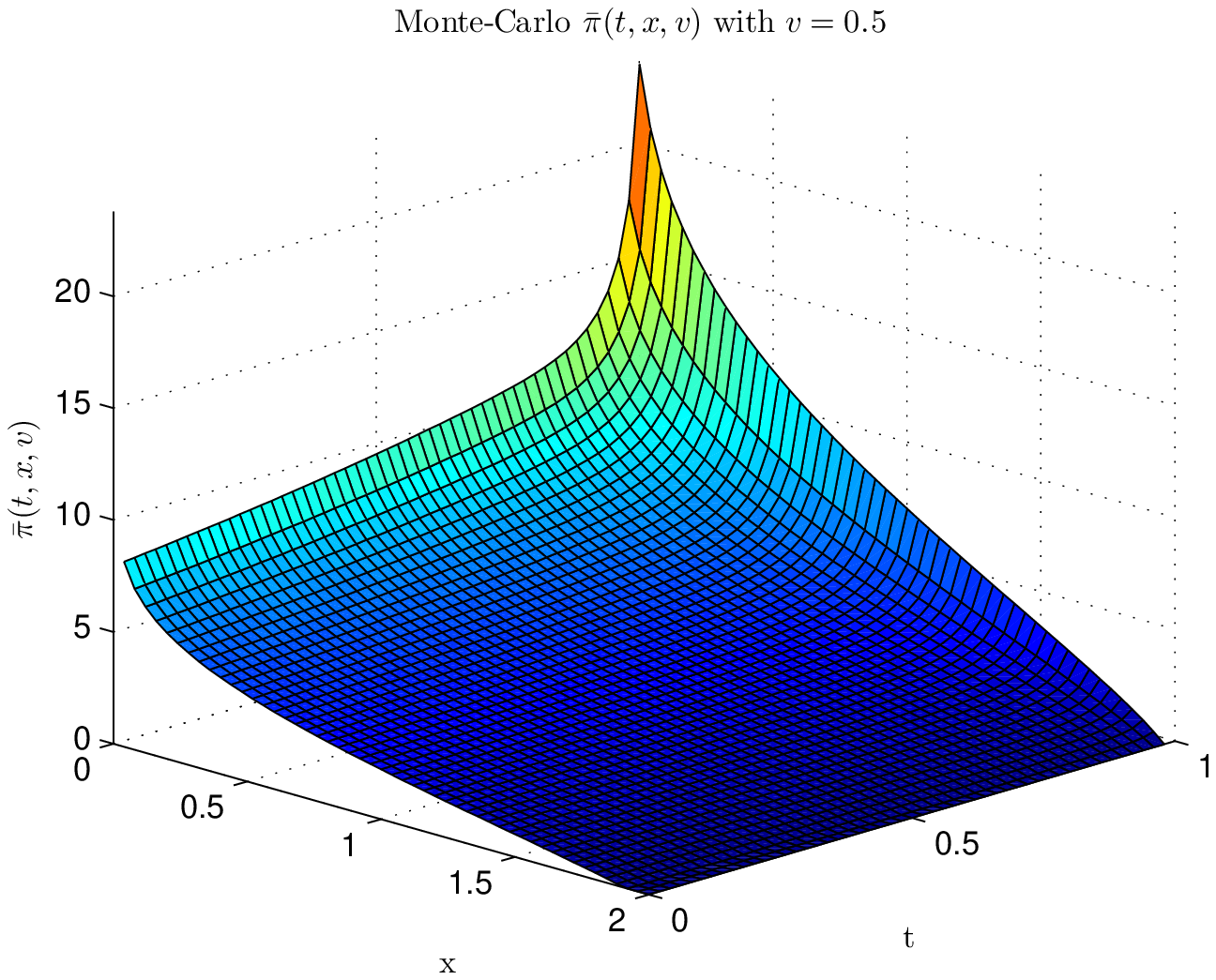}
\end{minipage}
}
\subfigure
{
 \begin{minipage}{1.85in}
 \centering
 \includegraphics[width=1.85in,height=1.85in]{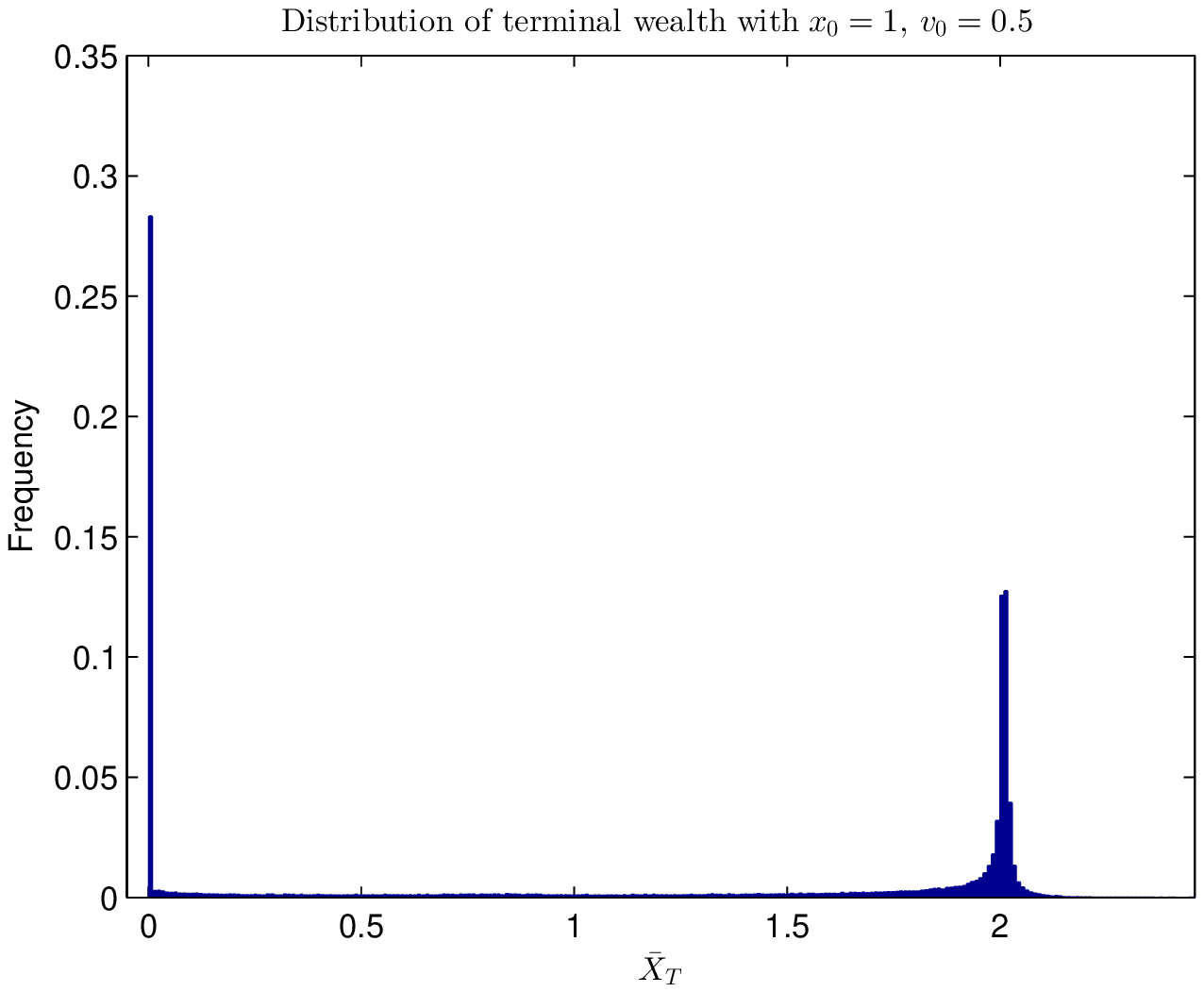}
\end{minipage}
}
\caption{3D and 2D figures for Yarri utility (Example \ref{exam:Yarri-vt-not-constant}). The left figure is the optimal strategies $\bar{\pi}(t,x,v)$ with initial wealth $x_0=1$. The middle figure is the optimal strategies $\bar{\pi}(t,x,v)$ with initial variance $v_0 = 0.5$. The right figure is the distribution of terminal wealth process.} %
\label{Yarri-pi-Csqrtv}
\end{figure}

\begin{exam}\label{exam:Yarri-robust-test}
In this example, we further test the robustness of the lower and upper bound methods for $\gamma_t = c\sqrt{v_t}$. The comparisons are carried out for the cases of sampling control $c$ for $1,5,10,20$ times uniformly distributed in $[-0.5,0.5]$ both for the lower and upper bounds. In Table \ref{Yarri-table-3}, we give the mean and standard deviation of the absolute and relative difference between the lower and upper bounds with randomly sampled parameters-sets: $10$ samples of $r$ from the uniform distribution on interval $[0.01,0.08]$, $\rho$ on $[-1,1]$, $\kappa$ on $[1,10]$, $\theta$ on $[0.01,1]$, $\xi$ on $[0.1,1]$, $A$ on $[0.1,1.5]$, $x_0=1$, $v_0=0.5$, and $T=1$. It is clear that the gap between the tight lower and upper bounds is very small, especially when the dual control $c$ is used. This shows that the algorithm is reliable and accurate.

\begin{table}[!htbp]
\caption{Mean and std of the absolute and relative difference between the lower and upper bounds for Yarri  utility (Example~\ref{exam:Yarri-robust-test}) with many randomly sampled parameters-sets.} 
\centering 
\begin{tabular}{l c c c c r}\label{Yarri-table-3}\\
\hline
Num $c$	&mean diff	            &std diff	            &mean rel-diff (\%)	    &std rel-diff (\%)	    &mean time\\
1	    &$6.9575\text{e}$$-3$	&$5.0882\text{e}$$-3$	&$5.2082\text{e}$$-1$	&$3.4214\text{e}$$-1$	&$1.22\text{e}$$+3$\\
5	    &$6.3116\text{e}$$-3$	&$4.7772\text{e}$$-3$	&$4.7119\text{e}$$-1$	&$3.2324\text{e}$$-1$	&$1.73\text{e}$$+3$\\
10	    &$6.0162\text{e}$$-3$	&$4.7874\text{e}$$-3$	&$4.4703\text{e}$$-1$	&$3.2578\text{e}$$-1$	&$2.37\text{e}$$+3$\\
20	    &$5.8369\text{e}$$-3$	&$4.5691\text{e}$$-3$	&$4.3418\text{e}$$-1$	&$3.1094\text{e}$$-1$	&$3.65\text{e}$$+3$\\
\hline
\end{tabular}
\end{table}
\end{exam}


\section{Conclusions}\label{conclusions}

In this paper we use the weak duality relation to construct the lower and upper bounds on the primal value function for  utility maximization under the Heston stochastic volatility model with general utilities. We propose a dual control Monte Carlo method to compute the bounds and suggest some simple forms of the dual control $\gamma_{t}$ which makes the bounds tighter and computation easier. In particular, if $\gamma$ is taken as  $\gamma_{t}=c(t)\sqrt{v_{t}}$ with $c$ being a piecewise constant function, the closed form upper bound can be obtained for a broad class of utilities (including power and non-HARA utilties), and the Fourier-Cosine formula can be used for the Yarri utility. The gap between the lower and upper bounds can be reduced if the number of sampling or the number of time pieces increases. Numerical examples show that the tight bounds can be derived with little computational cost.

\section*{Appendix A}

In the paper we need to solve a number of times the following system of equations:
$$   C'(t)= d_1 D(t) + d_2, \quad \underline t\leq t\leq \bar t $$
and
$$
D'(t)=a D^2(t) + b D(t) + \eta, \quad \underline t\leq t\leq \bar t,$$
with the terminal conditions $C(\bar t)=f_1$ and $D(\bar t)=f_2$, where all coefficients are constants.
Assume $b^2-4a\eta>0$ and $\frac{m_1}{m_2} \notin [e^{-k_1(\bar t- \underline t)},1]$,
where
\begin{equation*} \label{notations-power-result-2}
 k_1=\sqrt{b^2-4a\eta},\quad
m_1=\frac{-b-k_1}{2a}, \quad m_2=\frac{-b+k_1}{2a}.
\end{equation*}
 The assumption $b^2-4a\eta>0$ ensures $m_1$ and $m_2$ are  distinct real numbers.
We can first find $D$ by writing the equation   as
\begin{equation*}\label{ODE-Dt-1-new-form}
\frac{1}{a}\left(\frac{1}{D-m_1}-\frac{1}{D-m_2}\right)\frac{dD}{m_1-m_2}=dt.
\end{equation*}
Using the terminal condition and denote $k_2 = \frac{f_2-m_1}{f_2-m_2}$ we obtain the solution of (\ref{ODE-Dt-1-new-form}) as
\begin{equation}\label{Dt-2}
D(t) = \frac{m_1 - m_2}{1 - k_2 \exp[k_1(\bar t-t)]} + m_2,
\end{equation}
which leads to a closed-form formula for $D(t)$ on interval $[\underline t, \bar t]$.
As for $C(t)$, we have the following form
\begin{equation}\label{Ct-2}
C(t) = - \frac{d_1(m_1-m_2)}{k_1}\ln\left(\frac{k_2-1}{k_2-\exp[-k_1(\bar t - t)]}\right) - d_1m_2(\bar t - t) - d_2(\bar t - t) + f_1.
\end{equation}
The assumption $\frac{m_1}{m_2} \notin [e^{-k_1(\bar t - \underline t)},1]$ is to exclude the case of $\int_{\underline t}^{\bar t} D(s) ds$ being hypersingular integral.

\section*{Appendix B}
This part explains the main idea of COS method in \cite{Fang2008} and derives the formula for computing the upper bound for Yarri utility.

For a function supported on $[0,\pi]$, the cosine expansion reads
\begin{eqnarray*}
g(\theta) = \mathop{{\sum}'}\limits_{k=0}^{+\infty}\mathcal{A}_k \cdot \cos(k\theta)& \text{with} & \mathcal{A}_k =\frac{2}{\pi}\int_0^{\pi}g(\theta)\cos(k\theta)d\theta,
\end{eqnarray*}
where $\mathop{{\sum}'}$ indicates that the first term in the summation is weighted by one-half. For functions supported on any other finite interval, say $[\zeta_1,\zeta_2]\subset\mathbb{R}$, the Fourier-cosine series expansion can easily be obtained via a change of variables:
\begin{eqnarray*}
\theta:=\frac{x-\zeta_1}{\zeta_2-\zeta_1}\pi; \quad x=\frac{\zeta_2-\zeta_1}{\pi}\theta + \zeta_1.
\end{eqnarray*}
It then reads
\begin{eqnarray*}
g(x)=\mathop{{\sum}'}\limits_{k=0}^{+\infty}\mathcal{A}_k\cdot \cos\left( k\pi\frac{x-\zeta_1}{\zeta_2-\zeta_1}\right),
\end{eqnarray*}
with
\begin{eqnarray}\label{int-Ak}
\mathcal{A}_k = \frac{2}{\zeta_2-\zeta_1}\int_{\zeta_1}^{\zeta_2} g(x)\cos\left(k\pi\frac{x-\zeta_1}{\zeta_2-\zeta_1}\right)dx.
\end{eqnarray}
Suppose $[\zeta_1,\zeta_2]\subset \mathbb{R}$ is chosen such that the truncated integral approximates the infinite counterpart well, i.e.,
\begin{eqnarray}\label{CCF}
\phi_1(\omega):=\int_{\zeta_1}^{\zeta_2} e^{i\omega x}g(x)dx \approx \int_{\mathbb{R}}e^{i\omega x} g(x)dx = \phi(\omega).
\end{eqnarray}
Comparing equation \eqref{CCF} with the cosine series coefficients of $f(x)$ on $[\zeta_1,\zeta_2]$ in \eqref{int-Ak}, we find that
\begin{eqnarray*}
\mathcal{A}_k \equiv \frac{2}{\zeta_2-\zeta_1} \mathrm{Re}\left\{ \phi_1\left(\frac{k\pi}{\zeta_2-\zeta_1}\right)  \cdot \exp\left(-i\frac{k\zeta_1\pi}{\zeta_2-\zeta_1}\right)\right\},
\end{eqnarray*}
where $\mathrm{Re}\{\cdot\}$ denotes taking the real part of the argument. It then follows from \eqref{CCF} that $\mathcal{A}_k \approx \mathcal{F}_k$ with
\begin{eqnarray*}
\mathcal{F}_k \equiv \frac{2}{\zeta_2-\zeta_1} \mathrm{Re}\left\{ \phi\left(\frac{k\pi}{\zeta_2-\zeta_1}\right) \cdot \exp\left(-i\frac{k\zeta_1\pi}{\zeta_2-\zeta_1}\right)\right\}.
\end{eqnarray*}
We now replace $\mathcal{A}_k$ by $\mathcal{F}_k$ in the series expansion of $g(x)$ on $[\zeta_1,\zeta_2]$, i.e.,
\begin{eqnarray*}
g_1(x)= \mathop{{\sum}'}\limits_{k=0}^{+\infty} \mathcal{F}_k \cdot \cos\left(k\pi\frac{x-\zeta_1}{\zeta_2-\zeta_1}\right)
\end{eqnarray*}
and truncate the series summation such that
\begin{eqnarray*}
g_2(x)= \mathop{{\sum}'}\limits_{k=0}^{N-1} \mathcal{F}_k \cdot \cos\left(k\pi\frac{x-\zeta_1}{\zeta_2-\zeta_1}\right).
\end{eqnarray*}
For an option pricing problem expressed in \eqref{dual-val-func-Yaari}, we rewrite it in the following form
\begin{eqnarray}\label{Z0}
 \widetilde{\mathcal{Z}}(t,z,v)=E_{t,z,v}[L(1-e^{Z_{T}})^+] = \int_{\mathbb{R}}  \widetilde{\mathcal{Z}}(T,y) g(y|z,v)dy.
\end{eqnarray}
Since the density rapidly decays to zero as $y\rightarrow\pm\infty$ in \eqref{Z0}, we truncate the infinite integration range without loosing significant accuracy to $[\zeta_1,\zeta_2]\subset \mathbb{R}$, and  obtain approximation $ {\widetilde{\mathcal{Z}}^{(1)}}$:
\begin{eqnarray*}\label{Z1}
{\widetilde{\mathcal{Z}}^{(1)}}(t,z,v)=\int_{\zeta_1}^{\zeta_2} \widetilde{\mathcal{Z}}(T,y) g(y|z,v)dy.
\end{eqnarray*}
In the second step, since $g(y|z,v)$ is usually unknown whereas the characteristic function is, we replace the density by its cosine expansion in $y$,
\begin{eqnarray*}
g(y|z,v) = \mathop{{\sum}'}\limits_{k=0}^{+\infty} \mathcal{A}_k(z,v)\cos\left(k\pi\frac{y-\zeta_1}{\zeta_2-\zeta_1}\right)
\end{eqnarray*}
with
\begin{eqnarray*}
\mathcal{A}_k(z,v):=\frac{2}{\zeta_2-\zeta_1}\int_{\zeta_1}^{\zeta_2} g(y|z,v)\cos\left(k\pi\frac{y-\zeta_1}{\zeta_2-\zeta_1}\right)dy,
\end{eqnarray*}
which leads to
\begin{eqnarray*}
 {\widetilde{\mathcal{Z}}^{(1)}}(t,z,v) = \int_{\zeta_1}^{\zeta_2} \widetilde{\mathcal{Z}}(T,y) \mathop{{\sum}'}\limits_{k=0}^{\infty} \mathcal{A}_k(z,v)\cos\left(k\pi\frac{y-\zeta_1}{\zeta_2-\zeta_1}\right) dy.
\end{eqnarray*}
Interchanging the summation and integration, and inserting the definition
\begin{eqnarray*}
 \widetilde{\mathcal{Z}}_k: = \frac{2}{\zeta_2-\zeta_1}\int_{\zeta_1}^{\zeta_2} \widetilde{\mathcal{Z}}(T,y)\cos\left(k\pi\frac{y-\zeta_1}{\zeta_2-\zeta_1}\right) dy= \frac{2}{\zeta_{2}-\zeta_{1}}L [\psi_k(\zeta_{1},0)-\chi_k(\zeta_{1},0)],
\end{eqnarray*}
where
\begin{eqnarray*}
\psi_{k}(x_{1},x_{2}) &=&
\left\{ \begin{array}{ll}
\left[\sin\left(k\pi\frac{x_{2}-\zeta_{1}}{\zeta_{2}-\zeta_{1}}\right) -\sin\left(k\pi\frac{x_{1}-\zeta_{1}}{\zeta_{2}-\zeta_{1}}\right)\right]
\frac{\zeta_{2}-\zeta_{1}}{k\pi},& k\neq 0,\\
x_{2}-x_{1},& k=0,\\
\end{array}
\right.\\
\chi_{k}(x_{1},x_{2}) &=& \frac{1}{1+\left(\frac{k\pi}{\zeta_{2}-\zeta_{1}}\right)^{2}} \Big[ \cos\Big(k\pi \frac{x_{2}-\zeta_{1}}{\zeta_{2}-\zeta_{1}}\Big)e^{x_{2}} -\cos\Big(k\pi \frac{x_{1}-\zeta_{1}}{\zeta_{2}-\zeta_{1}}\Big) e^{x_{1}}\\
&& +\, \frac{k\pi}{\zeta_{2}-\zeta_{1}} \sin\Big(k\pi \frac{x_{2}-\zeta_{1}}{\zeta_{2}-\zeta_{1}}\Big)e^{x_{2}} -\frac{k\pi}{\zeta_{2}-\zeta_{1}} \sin
\Big(k\pi \frac{x_{1}-\zeta_{1}}{\zeta_{2}-\zeta_{1}}\Big) e^{x_{1}}\Big],
\end{eqnarray*}
we have
\begin{eqnarray*}
 {\widetilde{\mathcal{Z}}^{(1)}}(t,z,v) = \frac{1}{2}(\zeta_2-\zeta_1)\mathop{{\sum}'}\limits_{k=0}^{+\infty}\mathcal{A}_k(z,v) \widetilde{\mathcal{Z}}_k.
\end{eqnarray*}
Due to the rapid decay rate of these coefficients, we further truncate the series summation to obtain approximation ${\widetilde{\mathcal{Z}}^{(2)}}$:
\begin{eqnarray*}
{\widetilde{\mathcal{Z}}^{(2)}}(t,x,v) = \frac{1}{2}(\zeta_2-\zeta_1)\mathop{{\sum}'}\limits_{k=0}^{N-1}\mathcal{A}_k(z,v) \widetilde{\mathcal{Z}}_k.
\end{eqnarray*}
Approximating $\mathcal{A}_k(z,v)$ by $\mathcal{F}_k(z,v)$, we obtain
\begin{eqnarray*}
\widetilde{\mathcal{Z}}(t,x,v) \approx {\widetilde{\mathcal{Z}}^{(3)}}(t,x,v) =\mathop{{\sum}'}\limits_{k=0}^{N-1}\mathrm{Re}\left\{\phi\left(t,z,v; \frac{k\pi}{\zeta_2-\zeta_1}\right)e^{-ik\pi\frac{\zeta_1}{\zeta_2-\zeta_1}}\right\} \widetilde{\mathcal{Z}}_k.
\end{eqnarray*}
Define $\varphi(t,v;\omega)=e^{-i\omega z}\phi(t,z,v;\omega)$. This is the conditional characteristic function of $Z_T-Z_t= \ln Y_T-\ln Y_{t}$. Then we can obtain \eqref{Yaari-dual-val-function}.
To find $y^*$ for the \eqref{Yaari-upper-bound-2}, we need to solve the following equation
\begin{equation*}
\mathcal{Z}_y \approx \mathop{{\sum}} \limits_{k=1}^{N-1} \mathrm{Re} \left\{ \varphi\left(t,v;\frac{k\pi}{\zeta_{2}-\zeta_{1}}\right)e^{ik\pi\frac{\ln y-\zeta_{1}}{\zeta_{2}-\zeta_{1}}}\frac{ik\pi}{(\zeta_{2}-\zeta_{1})y} \right\}\widetilde{\mathcal{Z}}_k=-x.
\end{equation*}
To compute the feedback control \eqref{pi-expression}, we need the following derivative
\begin{eqnarray*}
&&\mathcal{Z}_{yy} \approx  -\mathop{{\sum}} \limits_{k=1}^{N-1} \mathrm{Re} \left\{ \varphi\left(t,v;\frac{k\pi}{\zeta_{2}-\zeta_{1}}\right)e^{ik\pi\frac{\ln y-\zeta_{1}}{\zeta_{2}-\zeta_{1}}}\frac{1}{y^2}\left[ \frac{k^2\pi^2}{(\zeta_{2}-\zeta_{1})^2} + \frac{ik\pi}{\zeta_{2}-\zeta_{1}}\right] \right\}\widetilde{\mathcal{Z}}_k,\\
&&\mathcal{Z}_{yv} \approx  \mathop{{\sum}} \limits_{k=1}^{N-1} \mathrm{Re} \left\{ \varphi\left(t,v;\frac{k\pi}{\zeta_{2}-\zeta_{1}}\right)\widetilde{D} \left(t;\frac{k\pi}{\zeta_{2}-\zeta_{1}}\right)
e^{ik\pi\frac{\ln y-\zeta_{1}}{\zeta_{2}-\zeta_{1}}} \frac{ik\pi}{(\zeta_{2}-\zeta_{1})y} \right\}\widetilde{\mathcal{Z}}_k.
\end{eqnarray*}
Finally, according to \cite{Fang2008}, we can choose the boundary of integral as
\begin{equation*}
[\zeta_1,\zeta_2]:= \left[c_1-L_1\sqrt{|c_2|},~ c_1+L_1\sqrt{|c_2|}\right],
\end{equation*}
where
\begin{eqnarray*}
c_n = \frac{1}{i^n}\frac{\partial^n \ln(\varphi(t,v;\omega))}{\partial \omega^n}|_{\omega=0},
\end{eqnarray*}
and $L_1$ is a constant chosen large enough to guarantee $\zeta_1<0<\zeta_2$. Cumulant $c_2$ may become negative for sets of Heston parameters that do not satisfy the Feller condition, i.e., $2\kappa\theta \geq \xi^2$. We therefore use the absolute value of $c_2$.

\end{document}